\documentclass[11pt, a4paper]{article}
\usepackage[english]{babel}
\usepackage{amssymb,amsmath,amsthm,bm,graphicx}
\usepackage{authblk}
\usepackage[authoryear]{natbib}
\usepackage{multirow}
\usepackage{mathtools}
\usepackage[]{algorithm2e}
\usepackage[dvipsnames]{xcolor}
\usepackage{hyperref}

\newtheorem{theorem}{Theorem}

\theoremstyle{definition}

\newtheorem{remark}{Remark}

\newcommand{\indic}{{\bm 1}}

\def\beqn{\begin{eqnarray*}}
\def\eeqn{\end{eqnarray*}}
\def\beq{\begin{eqnarray}}
\def\eeq{\end{eqnarray}}
\def\bm#1{\mbox{\boldmath{$#1$}}}

\begin{document}

\title{\Large\bfseries
Estimation of Local Degree Distributions via Local Weighted Averaging and Monte Carlo Cross-Validation
}
\author{\bfseries Paulo Serra\thanks{
This research partially took place while this author was a postdoctoral researcher at the Korteweg-de Vries Institute for Mathematics, of the University of Amsterdam, in Amsterdam, the Netherlands.} $^{,\dagger}$\,
and Michel Mandjes\thanks{The research for this paper is partly funded by the NWO Gravitation Programme N{\sc etworks}, Grant Number 024.002.003 (Serra, Mandjes), and an NWO Top Grant, Grant Number 613.001.352 (Mandjes).}\\ \vspace{-0.75em} {\rm Eindhoven University of Technology and University of Amsterdam}}
\date{\today}

\maketitle

\begin{abstract}
  Owing to their capability of  summarising interactions between elements of a system, networks have become a common type of data in many fields.
  As networks can be inhomogeneous, in that different regions of the network may exhibit different topologies,
  an important topic concerns their local properties.
  This paper focuses on the estimation of the local degree distribution of a vertex in an inhomogeneous network.
  The contributions are twofold: we propose an estimator based on local weighted averaging, and we set up a Monte Carlo cross-validation procedure to pick the parameters of this estimator.
  Under a specific modelling assumption we derive an oracle inequality that shows how the model parameters affect the precision of the estimator.
  We illustrate our method by several numerical experiments, on both real and synthetic data, showing in particular that the approach considerably improves upon the natural, empirical estimator.
\end{abstract}

{{\bf Keywords:}
local degree distribution,
local weighted averaging,
Monte Carlo cross-validation,
oracle inequality,
random connection model,
wireless ad-hoc networks.
}

\section[Introduction]{Introduction}\label{sec:introduction}

Networks are capable of summarising the interactions between  elements of a system. This explains why they have become    a common type of data in a variety of fields, such as statistics, computer science, the social sciences, physics, biology/genomics, epidemiology, transport and communications, and neuroscience; see e.g.\ \cite{kolaczyk2014statistical}.
An important feature of many networks encountered in these fields is that they are {\it heterogeneous}, in the sense that different regions of the network may have rather different topological properties~\cite{barabasi1999emergence}.
For instance in networks with an underlying community structure  heterogeneity
arises, with the heterogeneity becoming more pronounced in case  the communities are intrinsically different~\cite{girvan2002community}.
Other examples of heterogeneous networks include protein-protein interaction networks, gene regulatory networks, gas/electricity/water networks, brain networks, wireless networks, inter-bank networks, road networks, and networks used in recommender systems and security monitoring~\cite{barabasi2016network}.

Studies of heterogeneous networks typically focus on features like the degree distribution, the clustering coefficient, and the average path length, as these greatly determine the topology of the network~\cite{van2009random}.
It is clear, however, that in many situations  one is interested in a {\it local version} of such features (i.e., the features mentioned above, but then  {in relation to a specific vertex}, or a specific region of the network), rather than their global counterparts.
The main topic of this paper  is the estimation of the {\it local degree distribution}, i.e., the degree distribution of a fixed vertex of interest (henceforth called the {\it origin}).

It is important to note that
there are also computational reasons to focus on local properties of networks.
Most online social networks, for example, are nowadays made up of millions of vertices and tens of millions of edges, making sampling itself challenging~\cite{Yang:2012aa}.
Local features, however, are estimated using only local information and are therefore more tractable, and amenable to parallel computing.

\vspace{3mm}

In this paper we follow a model-based approach to the problem of estimating local degree distributions.
Random graphs are popular choices when it comes to modelling complex networks; cf.~\cite{goldenberg2010survey,newman2010networks,van2009random,wasserman1994social}.
Here, we model our network according to a {\it random connection model} (RCM), a model from continuum percolation; cf.~\cite{penrose1991continuum} and (or an overview on the subject) \cite{meester1996continuum}. We also refer to~\cite{serra2017dimension} for an application of such models to the estimation of the intrinsic dimension of high-dimensional datasets.
Associated with each vertex is a (latent) feature, each sampled independently in some (potentially infinite-dimensional) metric space;
given these features, edges are placed independently between each possible pair of vertices with a probability that is a non-increasing function of the distance between the features corresponding to the vertices.
This model is quite flexible, in the sense that it introduces  heterogeneity in a natural manner. At the same time, it captures several random graph models as particular cases.

The  algorithm we propose  is fairly simple.
Since we only observe the degree of the origin once, an empirical estimate of the degree distribution (or, equivalently, of the connection probability of the origin) is trivial, namely the degree of the origin divided by the number of other vertices. The question, however, is whether we can do better.
Now observe that, since the connection probability is a smooth function of the location of the vertex on the feature space, it is reasonable to expect that the connection probability of neighbours of the origin to be, with high probability, similar to the connection probability of the origin itself. The idea is  therefore to borrow information from the degrees of vertices in a neighbourhood of the origin.
The resulting estimator depends on two parameters:
a set of neighbours to borrow information from, and
a set of weights to balance the contributions of these neighbours to the estimate.
Following this procedure will reduce the variance of the estimator, but this comes at the cost of introducing a bias into the estimate as a consequence of the (potentially uninformative) degrees of the neighbours.
In this sense, our estimator falls in the context of locally weighted linear estimators~\cite{friedman2001elements,james2013introduction,wasserman2006all}.

A computational advantage of the proposed algorithm, is that the estimate can be evaluated in a recursive manner,
thus allowing for efficient implementation.
The first estimate is the trivial estimate based only on information about the origin;
each subsequent estimate includes information from an extra neighbour of previously included vertices.
This entails that the algorithm can be viewed as a stochastic approximation algorithm, where we see the introduction of a new neighbour into the estimator as a correction to the previous estimate; cf.~\cite{belitser2013online,kiefer1952stochastic,kushner2003stochastic,robbins1951stochastic}.

It is clear that a wide variety of possible weight sequences can be used. Our first result concerns an oracle inequality that bounds the error of our estimator as a function of (i)~the size of the neighbourhood (which can effectively be seen as a bandwidth parameter), (ii)~the weights for the degrees of the neighbouring vertices, and (iii)~the underlying smoothness of the function that determines the mean degrees of vertices in a neighbourhood of the origin; cf.~\cite{candes2006modern} for an overview on oracle inequalities.
Based on this inequality we can determine the optimal families of weights to be used in the estimator. %(as a function of the number of neighbours).

Our second result is a method for selecting the size of the neighbourhood, or, directly, the neighbours from which we borrow information.
This result is based on an estimate of the risk which is obtained via Monte Carlo cross-validation (MCCV)~\cite{xu2001monte}.
We show that this criterion provides an estimate of the risk of the estimator, as a function of the weights and of the size of the neighbourhood.
The estimate is inherently biased, but the bias is independent of the size of the neighbourhood (and of the weights that we consider).
The number of neighbours can then be picked as the minimiser of this criterion.
In this sense, the MCCV procedure also provides a criterion for early stopping of our recursive estimator; cf.~\cite{Raskutti:2014:ESN:2627435.2627446,MR2327601}.
In this way we have an estimator that adapts to unknown features of the underlying network; cf.~\cite{tsybakov2009introduction} for an overview on the principles of adaptive estimation.

We have performed a set of illustrative numerical simulations, focusing on estimating the connection probability of the origin.
These in the first place show how our approach considerably improves upon the empirical estimator  by  choosing the number of neighbours in an appropriate way (for various choices of the set of weights).
In the second place, we demonstrate that this number of neighbours tends to grow slowly with the number of vertices in the network, as long as the average degree of the origin does not grow too quickly.
We show how the minimiser of the MCCV criterion quickly stabilises even after a small number of replications, and how choosing the number of neighbours according to this criterion leads to an improvement over the empirical estimator.

In addition, we  apply the estimator to real data and simulated data.
We apply our method to estimate the probability that a specific vertex in the network establishes edges with other vertices.
We show that an appropriate choice of the weights, combined with our MCCV procedure, leads to considerable improvement over the empirical estimator.
We do this also under different choices for the distribution of the features that determine the topology of the network.
Furthermore, we consider an application to wireless ad-hoc networks~\cite{haas2002wireless}, where we investigate the minimal number of vertices in the network so as to ensure that an isolated node can access some other node in the network with a given probability.

This paper is structured as follows.
We introduce a model for heterogeneous networks in Section~\ref{sec:model}.
In Section~\ref{sec:estimator} we define our estimator.
We present in Section~\ref{sec:oracle_bound} an oracle bound on the error of the estimator as a function of its parameters.
In Section~\ref{sec:choice_of_weights_and_k} we outline a Monte Carlo cross-validation procedure to pick the parameters of the estimator in a data-driven way.
Although we illustrate the use of our estimator throughout the paper\footnote{The graphs in Figures~\ref{fig:SBM_example_1_graph}, \ref{fig:world_cities_graph}, and~\ref{fig:SIM_example_graph}, as well as the implementation of our estimator make use of R's   \emph{igraph} package~\cite{csardi2006igraph}.},
in Section~\ref{sec:application} we apply our estimator to a real dataset, whereas Section~\ref{sec:simulations} contains numerical experiments with simulated datasets.
Concluding remarks can be found in
Section~\ref{sec:conclusions}.
A possible direction for future work is discussed in Section~\ref{sec:future}.
All proofs can be found in Appendix~\ref{apx:proofs}.

% \newpage
%%%%%%%%%%%%%%%%%%%%%%%%%%%%%%%%%%%%%%%%%%%%%%%%%%
\section[A Model for Heterogeneous Networks]{A Model for Heterogeneous Networks}\label{sec:model}

As mentioned in the introduction, we model our heterogeneous network by a random connection model (RCM);
in this section we provide a formal desription.
To this end, let $X_1, \dots, X_n$, be \emph{features} sampled i.i.d.\ from some distribution $F$ on $\mathbb{R}^d$.
In addition, consider a `probe' feature $X_0=x$, where $x\in\mathbb{R}^d$ is of our choosing and fixed.
Abbreviate $\bm X = (X_0,X_1, \dots, X_n)$.
%These have the interpretation of being features for the underlying graph.
Let $\rho:\mathbb{R}_{0,+}\to[0,1]$ be a non-negative \emph{connection function}, which we allow to depend on $n$.
This function will typically have bounded support, or it will be such that $\rho(x)$ converges relatively fast to $0$ as $x\to\infty$.

A random graph is sampled from the model in the following way.
We first sample features $\bm X$, and given those we sample a random probability matrix $\bm P = (P_{i, j})_{i, j=0, \dots, n}$, where
\begin{equation}\label{def:P_i_j}
P_{i, j}:=\rho\big(\|X_i-X_j\|\big), \quad i,j = 0,\dots, n,\, \:\mbox{with}\:\:i\neq j, \qquad P_{i,i} = 0.
\end{equation}
Given $\bm P$ we sample an adjacency matrix $\bm A = (A_{i, j})_{i,j=0,\dots,n}$, by first sampling i.i.d.\ $\epsilon_{i, j}\sim U[0,1)$, $i,j=0,\dots,n$, and setting
\begin{equation}\label{def:A_i_j}
A_{i, j} := \indic\{P_{i, j}>\epsilon_{i, j}\},  \quad i,j = 0,\dots, n\, \:\mbox{with}\:\:i\neq j,
\end{equation}
so that (almost surely) $A_{i,i} = 0$, for $i=0,\dots,n$, and (in self-evident notation)
\begin{equation}\label{eq:A_i_j_distribution}
A_{i, j}\,|(X_i,X_j) \stackrel{\rm indep.}{\sim} {\rm Bin}(1,P_{i, j}), \quad i,j = 0,\dots, n\,\, \:\mbox{with}\:\:i\neq j.
\end{equation}
The random adjacency matrix $\bm A$ specifies a random directed graph $G$ on the vertices $\{0,1,\dots,n\}$; one could also modify the procedure to produce undirected graphs.

We refer to vertex $0$ as the \emph{origin}.
This vertex can be a vertex of our choosing that is already in the network, or it can be a vertex that is artificially introduced into the network to gain local information about it.

Consider the following function on $\mathbb{R}^d$:
\begin{equation}\label{def:connection_probability}
p(x) = \mathbb{E}\,\rho(\|X-x\|), \quad X\sim F,
\end{equation}
which we call the \emph{local connection probability} (at $X_0=x$).
Denote by $B_i$ the out-degree of vertex $i$, i.e.,
\begin{equation}\label{def:B_i}
B_i = \sum_{\substack{j=0\\ j\neq i}}^n A_{i, j}, \quad i=0, \dots, n,
\end{equation}
such that for any $i=0,\dots,n$,
\begin{equation}\label{def:B_i_given_X_i}
B_i\,|\,X_i \sim {\rm Bin}\big(n, p(X_i)\big),
\end{equation}
and in the specific case of the origin,
\begin{equation}\label{def:B_0}
B_0\,|\,X_0 = B_0 \sim {\rm Bin}\big(n, p(x)\big).
\end{equation}
Our objective is to estimate the local degree distribution (at $x$), i.e., the distribution of $B_0$;
equivalently, we want to estimate the local connection probability at $x$, $p(x)$.
We note that $p(x)$ may depend on $n$.

%%%%%%%%%%%%%%%%%%%%%%%%%%%%%%%%%%%%%%%%%%%%%%%%%%
\section[An Estimator for the Local Connection Probability]{An Estimator for the Local Connection Probability}\label{sec:estimator}

In this section we propose a family of estimators for the local connection probability $p(x)$.
For now the feature $X_0=x$ associated with the origin $0$ is fixed, so we often omit $x$ from the notation.
The natural, {\it empirical estimator} is
\begin{equation}\label{def:hat_p_0}
\hat p_0 := \frac{B_0}{n},
\end{equation}
This estimator is unbiased and
always available to us. In particular it satisfies the central limit theorem \footnote{See Appendix~\ref{apx:proofs:CLT}.}
\begin{equation}\label{eq:asymptotics_p_0}
\sqrt{\frac{n\, p(x)}{1-p(x)}}\left\{\frac{\hat p_0}{p(x)}-1\right\} \stackrel{\rm d}{\longrightarrow} N(0,1),
\end{equation}
as long as $n\,p(x)\to\infty$
(recalling  that we allowed $p(x)$ to depend on $n$).

For the problem of estimating a local connection probability to be meaningful, the underlying graph should be (at least locally) fairly sparse so that the average degree of a vertex is not too large.
The origin, in particular, should satisfy this property, so that the typical case should be that $n\, p(x)$ is either constant or grows slowly.
The quantity $n\, p(x)$, which is the expected degree of the origin, determines the rate for the empirical estimator.
The conclusion from the above is that one should expect the empirical estimator to have low precision. As a consequence, there is a lot to be gained from trying to improve upon it.

Having~\eqref{def:B_i_given_X_i} and~\eqref{def:B_0} in mind, if $p(x)$ were an arbitrary (i.e., not necessarily smooth) function of $x$, then one should not expect to be able to improve upon~\eqref{def:hat_p_0}.
However, if the model described in the previous section is appropriate for the network at hand, then~\eqref{def:connection_probability} suggests that $p(x)$ should be a smooth function of $x$:
if $\rho$ is a continuous function and $F$ is a continuous distribution, we expect small changes in $x$ to have a small impact on $p(x)$. %\footnote{\tiny\tt [MM] This sentence I don't get. Do you mean that you require that $\rho$ is smooth? Or that $F(\cdot)$ is continuous? [PS] I expand the sentence to clarify in the direction that you mention.}.
In this case, it is reasonable to try and improve upon~\eqref{def:hat_p_0} by borrowing information from neighbours of the origin.

Denote by $\delta(i,j)$ the geodesic distance between vertices $i$ and $j$, and let $V_k(i)$ be the collection of vertices that are within a distance $k$ of vertex $i$, i.e.,
\begin{equation}\label{def:V_k_of_i}
V_k(i) := \{j: \delta(i,j)\le k\}.
\end{equation}
We abbreviate $V_k=V_k(0)$.
Further, we also define $V_{-1}(i) := \emptyset$, for $i=0,\dots,n$.
A natural estimator of $p(x)$, that incorporates information from the neighbours of the origin (up to distance $k$), is
\begin{equation}\label{def:hat_p_k}
\hat p_k :=
\frac1{\sum_{\ell=0}^k w_\ell} \sum_{\ell=0}^k \frac{w_\ell}{|V_\ell\backslash V_{\ell-1}|}\sum_{i\in V_\ell\backslash V_{\ell-1}}\,\frac{B_i}n.
\end{equation}
Here $\{w_\ell\}_{\ell=0,\dots,n}$ (eventually depending on $n$) is a sequence of non-negative weights that specify the contribution of vertices based on their distance to the origin.
The estimator is also parameterised by the distance $k$;
note that~\eqref{def:hat_p_0} is recuperated from~\eqref{def:hat_p_k} by setting $k=0$ (i.e., by ignoring all information from vertices other than the origin), so the family of estimators $\hat p_k$ generalises the empirical estimator.

In the sequel we fix, without loss of generality, $w_0:=1$.
For now, the weights $w_\ell$ and the geodesic distance $k$ are left unspecified; one of the objectives of this paper is to determine appropriate choices for $w_\ell$ and   $k$.

\begin{remark}{\em
Our estimator is a weighed average of the degrees of all vertices that are close to the origin.
This closeness is measured in terms of the geodesic distance $k$ to the origin.
Alternatively, one can also average over the $m$ vertices that are closest  to the origin;
we denote this estimator as $\check p_m$.
%\footnote{\tt\tiny [MM] I think we have to avoid ambiguous notation -- one cannot use $\hat p_m$ for both. That's the reason why I introduced $\check p_m$. Please check whether I did this consistently. [PS] I agree; I double checked to see if the right notation is being used everywhere.}
Note that the estimates $\hat p_k$ simply correspond to a random subsequence of the estimates $\check p_m$.
The analysis of both estimators is fully analogous, but the notation becomes rather cumbersome when considering the estimator $\check p_m$.
Because of this, we frame our theoretical analysis of the properties of the estimator in terms of the \emph{neighbourhood size} $k$.
However, for the numerical implementation we use the estimator in terms of the \emph{number of neighbours} $m$.
This is because the effect of changing $m$ is much smaller than the effect of changing $k$, thus substantially improving the numerical stability of the implementation.} \hfill$\Diamond$
\end{remark}

The choice of the weight sequence  $\{w_\ell\}$ is rather important, as it must somehow   compensate for the speed at which the connection probability changes as one moves away from the origin.
The resulting estimates $\hat p_k$ should stabilise fast enough as $k$ grows, so as to capture  the local connection probability.

\vspace{3mm}

At different points in this paper we use the stochastic block model (SBM) to illustrate the behaviour of our estimator.
Although technically  RCM does not cover SBM, RCM can be used to model graphs stemming from SBM.
A simple way of doing this is to set $F$ to a mixture of, say, Gaussian (or uniform) distributions.
Then the number of mixture components control the number of communities;
the mixture weights controls the relative sizes of the different communities;
relative distances between location parameters for the mixture components control the probability of inter-community edges being formed;
scale parameters of each mixture component control the probability of intra-community edges being formed.
SBM is quite useful to illustrate our estimator since the underlying community structure of the graph is a good example of the local nature of the connection probability.
We also use it here for convenience, as in the SBM context the connection probability allows for easy calculation.

To gain  insight into the impact that the weight sequence $\{w_\ell\}$  has on~\eqref{def:hat_p_k}, we consider three special cases; these specific choices will be motivated in Section~\ref{sec:oracle_bound}. Define for $\ell = 1, \dots, n$ and some small $0<\gamma\ll1$ (depending on $n$),
\begin{equation}\label{eq:weight_examples}
w_\ell^{(1)} := |V_\ell\backslash V_{\ell-1}|, \qquad
w_\ell^{(2)} \equiv 1, \qquad
w_\ell^{(3)} := \gamma(1-\gamma)^{-\ell}
,\qquad
\end{equation}
and $w_0^{(i)}:=1$. Note that the weights $w_\ell^{(2)}$  are constant in $\ell=1, \dots, n$.
The first set of weights leads to a na\"{\i}ve estimator that satisfies
\[
n\cdot\hat p_k^{(1)} =
\frac1{|V_k|} \sum_{i\in V_k} B_i, \quad k=0,\dots,n
\]
(adopting the notation that $\hat p_k^{(i)}$ is the estimator based on the weights $\{w_\ell^{(i)}\}$).
The second set of weights weighs each annulus of neighbours of the origin equally.
The third set of weights down-weighs all vertices other than the origin since $\gamma$ is small.

We computed the estimates $\check p_m$ corresponding to these three weight sequences for a graph sampled from an SBM.
There were $n=50$ vertices, split into $3$ communities with respectively $c_1=10$, $c_2=25$, and $c_3=15$ vertices each.
The intra-community connection probabilities $p_i$, $i=1,2,3$, were such that $n\cdot p_i$ was $15$, $25$, and $15$, respectively, and the inter-community connection probability $q$ was such that $n\cdot q=0.5$.
In Figure~\ref{fig:SBM_example_1_graph} we show the network graph.

\begin{figure}[!htb]
\centering
\includegraphics[trim={0cm 0cm 0cm 0cm}, width=0.75\textwidth]{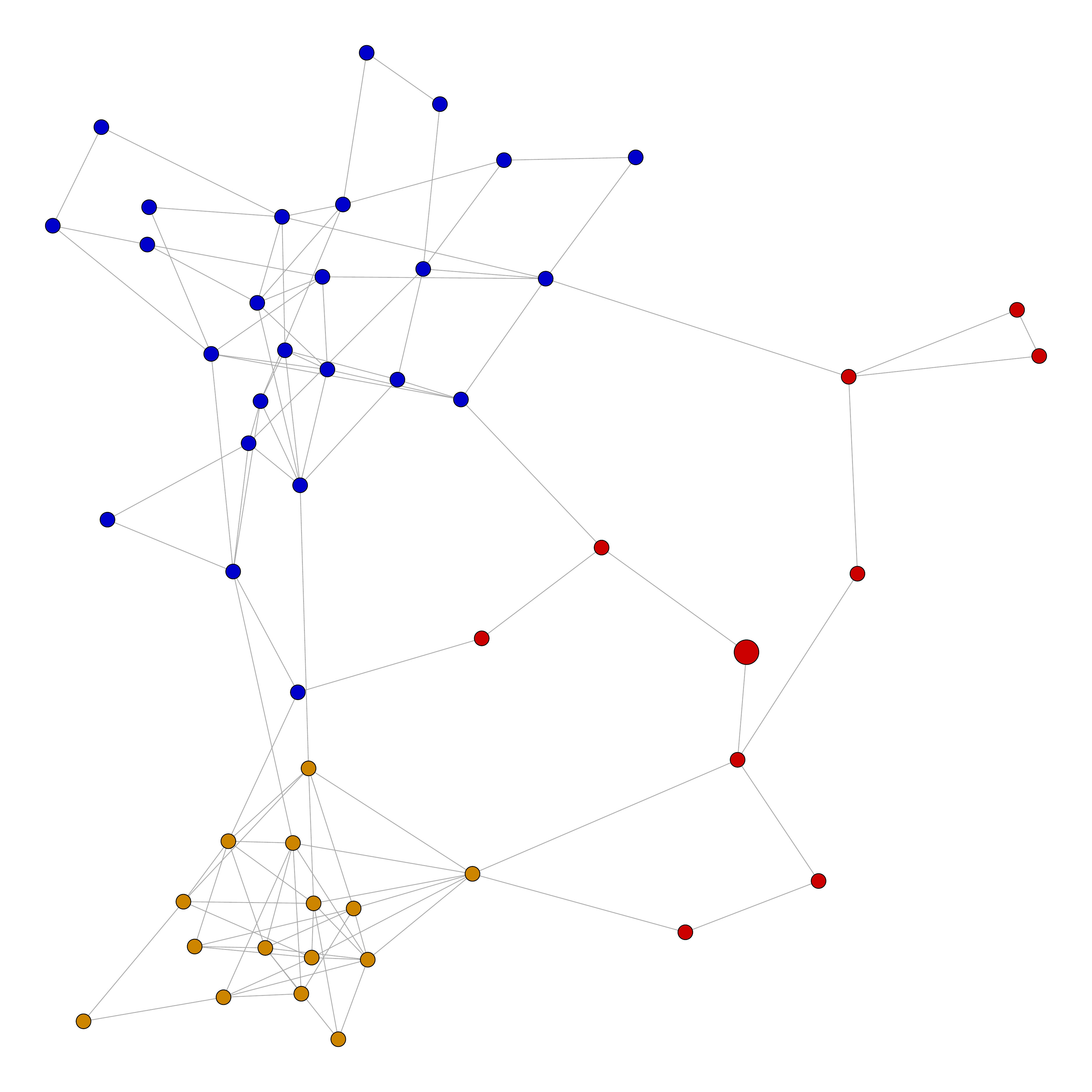}
\caption{\footnotesize
Example of a network sampled from the stochastic block model.
There are three communities corresponding to the three different colours;
the origin vertex is the larger vertex.
}\label{fig:SBM_example_1_graph}
\end{figure}

The connection probability is computed from the parameters of the SBM.
Since the origin was chosen to belong to the first community, $p(x)$ satisfies
\begin{equation}\label{eq:SBM_example_1_truth}
n\cdot p(x) = (c_1-1)\cdot p_1 + (n-c_1)\cdot q.
\end{equation}
This leads to $p(x) = (9 \cdot 0.3+ 40\cdot 0.01)/50 = 0.062$.
In Figure~\ref{fig:SBM_example_1_estimates} we show the corresponding estimates of $p(x)$ as a function of the number of neighbours $m$ (in the weights in~\eqref{eq:weight_examples}).

\begin{figure}[!htb]
\centering
\includegraphics[trim={0cm 0cm 0cm 0cm}, width=0.99\textwidth]{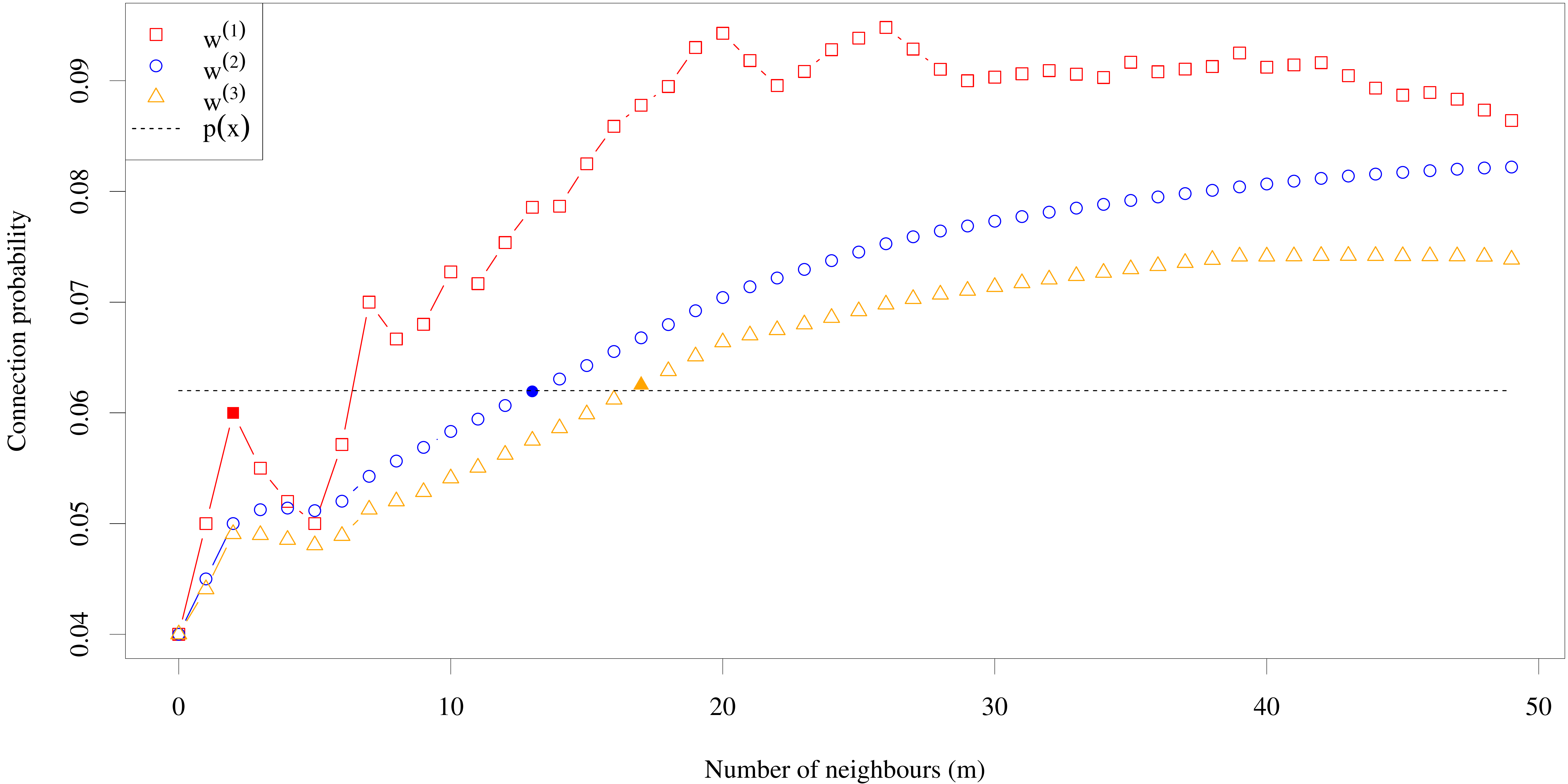}
\caption{\footnotesize
Estimates of the local connection probability for the network represented in Figure~\ref{fig:SBM_example_1_graph} corresponding to the the weight sequences in~\eqref{eq:weight_examples}, as a function of the number of neighbours $m$.
The underlying connection probability of the origin is marked with a dashed line.
For each sequence of estimates, the best estimate of the connection probability is marked in full.
}\label{fig:SBM_example_1_estimates}
\end{figure}

A few conclusions can be immediately drawn from Figure~\ref{fig:SBM_example_1_estimates}.
The empirical estimate $\hat p_0$ (corresponding to $m=0$ or $k=0$) can be substantially improved by using the weight sequences in~\eqref{eq:weight_examples}; it, however, requires that a (non-trivial) choice is made for the parameter $m$.
(This nontrivial choice is termed an \emph{oracle}, since it requires knownledge of the unknown $p(x)$.)
The choice of the weights has a strong impact on the corresponding sequences of estimates, and therefore on what the best choice for $m$ is.
Other than that, although the estimators behave differently, they all seem capable of producing estimates with similar accuracy.

The curves in Figure~\ref{fig:SBM_example_1_estimates} are specific for the network at hand. To have a clearer picture of how the different choices of weights affect the estimates, we independently simulated $10^5$ networks distributed like the one in Figure~\ref{fig:SBM_example_1_graph}, and estimated the corresponding mean squared errors (by averaging the squared residuals $m\mapsto\{\check p_m-p(x)\}^2$ over these networks).
Figure~\ref{fig:SBM_example_1_MSE} displays the results.

\begin{figure}[!htb]
\centering
\includegraphics[trim={0cm 0cm 0cm 0cm}, width=0.99\textwidth]{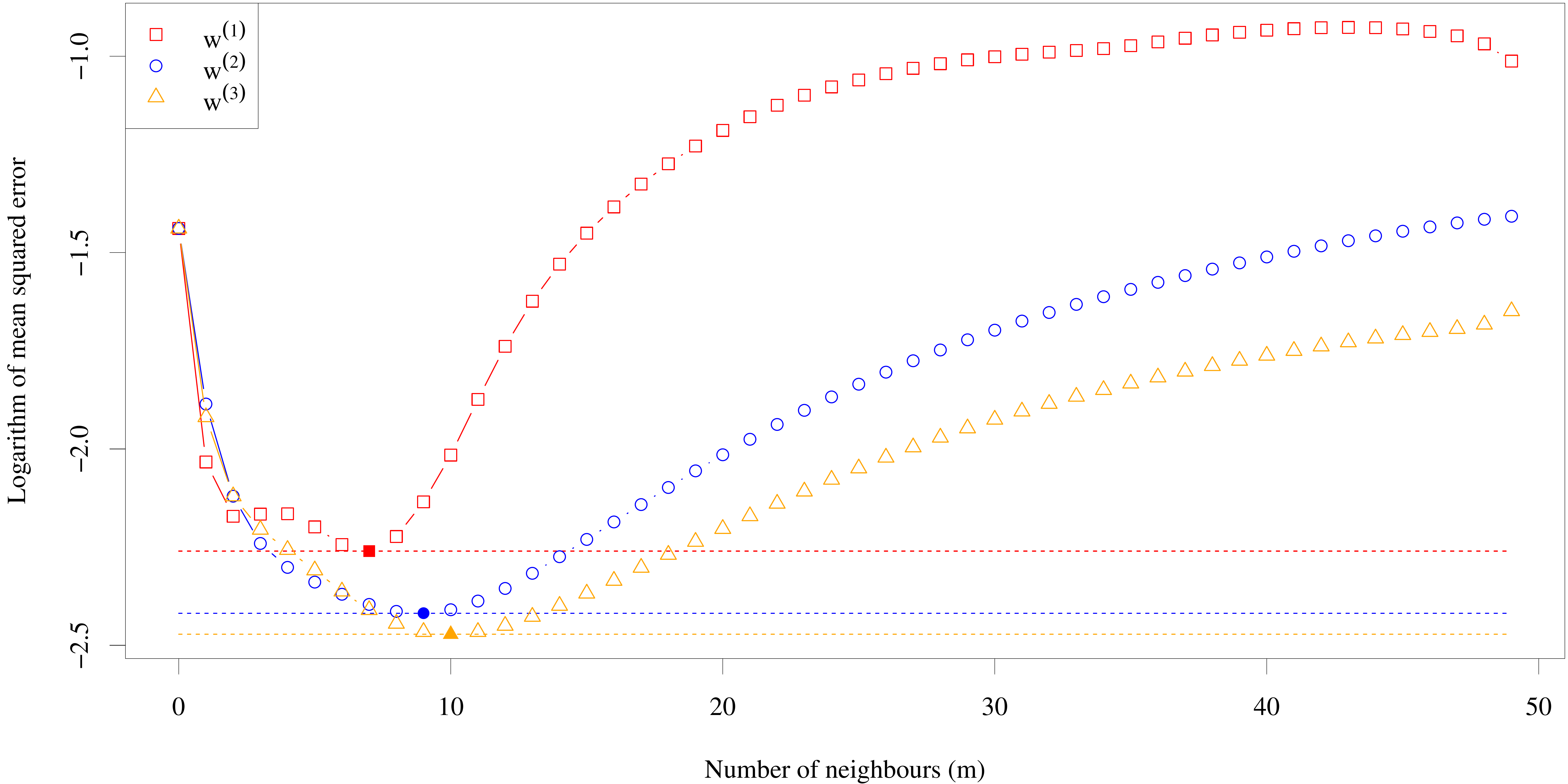}
\caption{\footnotesize
Logarithm of the simulated MSE of the estimators corresponding to the weight sequences in~\eqref{eq:weight_examples}.
The horizontal dashed lines mark the smallest error attained by each of the estimators.
For each of the weight sequences, the MSE corresponding to the best choice of $m$ is marked in full.
}\label{fig:SBM_example_1_MSE}
\end{figure}

Figure~\ref{fig:SBM_example_1_MSE} suggests that, on average, for the instance considered the weights $w^{(2)}$ and $w^{(3)}$
perform somewhat better than the weights $w^{(1)}$ (where it is noted that there is a small difference between the former two).
In this context, `better performance' means that the corresponding estimators are capable of attaining a smaller MSE (for an appropriate choice of the parameter $m$).

To exemplify how the optimal choice of the parameter $m$ is affected by the sample size $n$, we consider a sequence of graphs sampled from the SBM.
For each $n\in\{1+i\cdot 100, i=1,\dots,10\}$ we consider $3$ communities with respectively $c_{n,i}$, $i=1,2,3$, vertices with
\[
c_n := (c_{n,1},c_{n,2},c_{n,3}) =
\big([n/5],\, [n/2],\, n-c_{n,1}-c_{n,2}\big),
\]
where $[\,\cdot\,]$ denotes rounding to the closest integer.
This choice ensures that the relative sizes of the communities remains nearly constant as $n$ increases.
The probability that two vertices, both belonging to community $i$, establish an edge is $p_{n,i}$, with
\[
p_n = (p_{n,1},p_{n,2},p_{n,3}) = 1\wedge
\Big(15\cdot \frac{\log n}{n},\, 10\cdot \frac{\log n}{n},\, 20\cdot \frac{\log n}{n}\Big),
\]
where the minimum with $1$ should be understood entry-wise.
These probabilities ensure that the number of edges between each vertex and other vertices in the same community is proportional to $\log n$.
The probability that two vertices from different communities establish an edge is $q_n=1/n$;
this way, the average number of edges that each individual vertex establishes with vertices from other communities remains constant as $n$ grows.
The example from Figure~\ref{fig:SBM_example_1_graph} corresponds to setting $n=50$.

For these choices, it is readily checked that the underlying local connection probability is
\[
p(x) = p_n(x) = 3 \cdot \frac{\log n}n\big(1+o(1)\big).
\]
In particular, we see that the choices above put us in a regime where the empirical estimator $\hat p_0$ is consistent (with a rate of the order $\sqrt{\log n}$) in the sense of~\eqref{eq:asymptotics_p_0}.
Nonetheless, graphs sampled from the SBM with these parameters will still be relatively sparse so that the eccentricity of the origin is large.
This means that it is possible to pick $m$ to be large (so as to reduce the variance of the estimator) without the estimator degenerating.

\begin{figure}[!htb]
\centering
\includegraphics[trim={0cm 0cm 0cm 0cm}, width=0.75\textwidth]{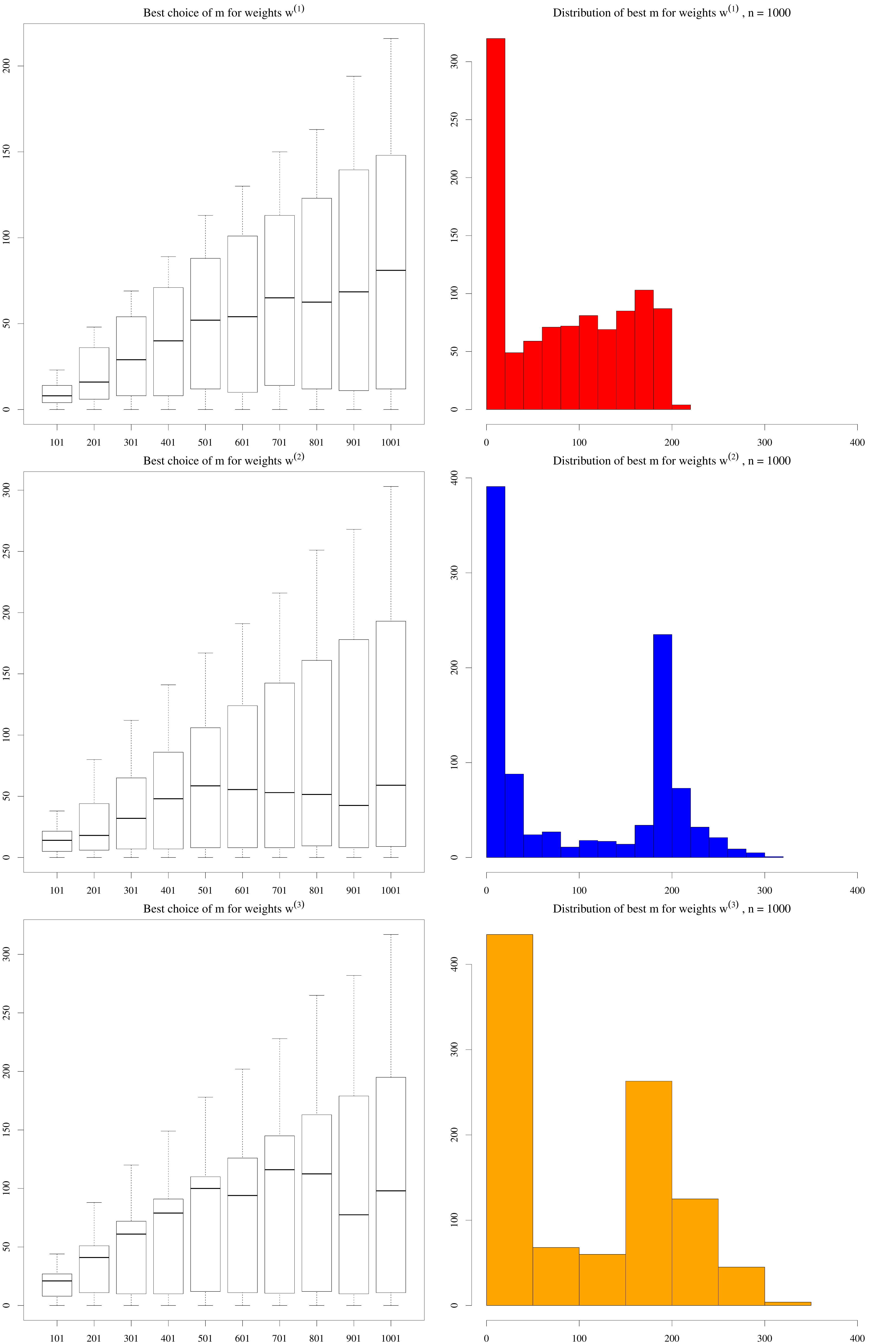}
\caption{\footnotesize
Best choice (oracle) for the number of neighbours $m$ for different sample sizes.
Each row of plots corresponds to a different weight sequence.
The boxplots in the panels on the left correspond to the best value for $m$ as far as optimising the MSE is concerned, for different choices of $n$.
In the panels on the right, each barplot shows the distribution of the $m$ for the largest value of $n$ that was considered.
}\label{fig:SBM_example_1_best_k}
\end{figure}

For each sequence of weights $w^{(i)}$, $i=1,2,3$, and each value of $n\in\{1+i\cdot100,i=1,\dots,10\}$ we simulated $10^3$ networks from the SBM with the aforementioned parameters.
Figure~\ref{fig:SBM_example_1_best_k} shows the results of this experiment.
The boxplots on the left show that the average values of the oracle
for $m$ seem to somewhat stabilise as $n$ grows, although there is quite a lot of variability.
(Note that the oracle is random since it depends on the underlying network.)
The histograms on the right show the distribution of the optimal $m$ for the three sets of weights corresponding to the largest sample size ($1000$, that is).
For the weights $w^{(1)}$ the optimal choice of $m$ tends to be rather small.
For the other two sets of weights the distribution of the optimal $m$ turns out  to be bimodal;
the first peak should reflect the intuitive idea of averaging over neighbours in the same community.
However, for each particular network, the bias of the estimator is in principle arbitrary, so that there could be other choices of $m$ (that do not necessarily conform to averaging over neighbours) that produce more accurate estimates.

\begin{figure}
%[!htb]
\centering
\includegraphics[trim={0cm 0cm 0cm 0cm}, width=0.66\textwidth]{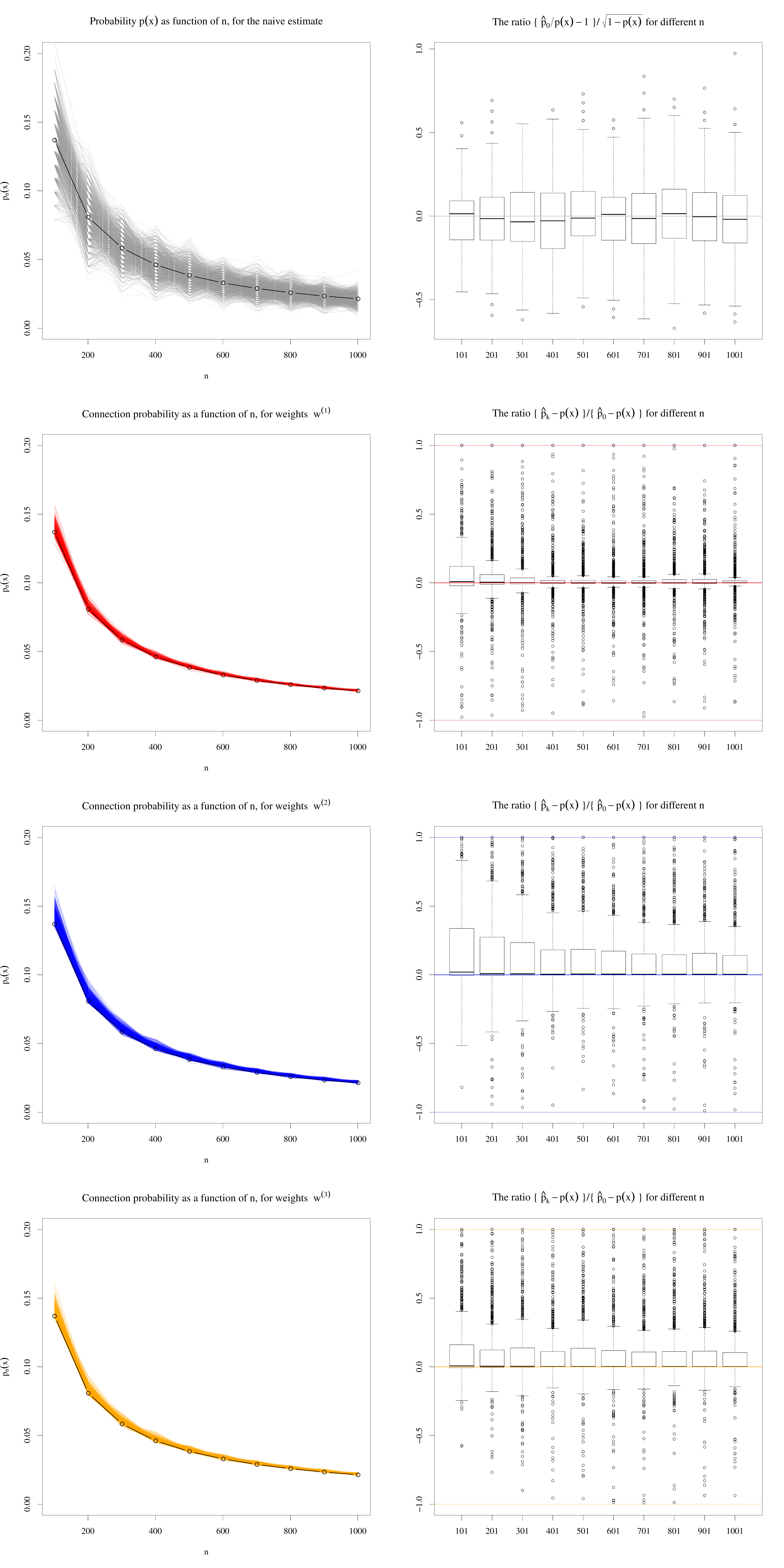}
\caption{\footnotesize
Estimates and errors corresponding to the best (oracle) choices of the number of neighbours $m$ for different sample sizes.
The left plots show the true connection probability (in black) for the different sample sizes, and the estimates for the empirical estimator and the three sets of weights in respectively grey, red, blue and orange.
On the right side we have boxplots of the errors for the different sample sizes.
The top one corresponds to the rescaled error of the empirical estimator, and the other three correspond to the ratio of the error estimates for each set of weights (for the oracle choice of $m$), divided by the error of the empirical estimator.
}\label{fig:SBM_example_1_best_ests}
\end{figure}

We now compare, for the chosen values of $n$ the estimates with the corresponding true values of the local connection probability.
%\clearpage
The error of the empirical estimator is, as suggested by~\eqref{eq:asymptotics_p_0}, under control.
It is already clear, visually, that any of the weight sequences leads to considerably better estimates of the connection probability.
Particularly the na\"{\i}ve estimator does quite well in this case.
This is probably due to the fact that the connection probability is constant in a neighbourhood of the origin so that it is less meaningful to weigh  the degrees of vertices based on their distance to the origin.

To quantify the effect of the weights and of the size of the neighbourhood on the estimator, we present in Section~\ref{sec:oracle_bound} an  oracle inequality for the MSE of our estimator~\eqref{def:hat_p_k}.

%%%%%%%%%%%%%%%%%%%%%%%%%%%%%%%%%%%%%%%%%%%%%%%%%%
\section[An Oracle Inequality]{An Oracle Inequality}\label{sec:oracle_bound}

To understand how the weights $w_l$, and the parameter $k$ should be selected so as to minimise the MSE of the estimator~\eqref{def:hat_p_k}, we present in this section an oracle inequality.
This inequality bounds the MSE of the estimator in terms of the model and of the parameters of the estimator.

If we take out the last term of the sum in~\eqref{def:hat_p_k} and scale appropriately, then
we find the recursion
\begin{equation}\label{def:recursion_one_step_estimator}
\hat p_{k+1} =
\hat p_k + \frac{w_{k+1}}{\sum_{l=0}^{k+1} w_l}\Bigg( \frac1{n|V_{k+1}\backslash V_{k}|}\sum_{i\in V_{k+1}\backslash V_{k}}B_i - \hat p_k \Bigg).
\end{equation}
This recursion can be written as $\hat p_{k+1} = \hat p_k + \gamma_k\cdot G_k$, where
\begin{equation}\label{def:recursion_weight_and_increment}
\gamma_k := \frac{w_{k+1}}{\sum_{l=0}^{k+1}w_l}, \quad
G_k := \frac1{n|V_{k+1}\backslash V_{k}|}\sum_{i\in V_{k+1}\backslash V_{k}}B_i - \hat p_k,
\end{equation}
$k = -1, \dots,n-1$.
Algorithms of this form are known as {\it stochastic approximation algorithms}~\cite{kiefer1952stochastic,kushner2003stochastic,robbins1951stochastic}.
The idea is that $G_k$ should behave like the average $p(X_i)-\hat p_k$ over vertices $i$ at a graph distance $k$ of the origin.
If $k$ is small enough so that on average $p(X_i)-p(x)$ is small, then $G_k$ should be close to $p(x)-\hat p_k$.
If this is the case, then the recursion~\eqref{def:recursion_one_step_estimator} moves each estimate $\hat p_k$ towards $p(x)$.
The sequence $\gamma_k$ is typically chosen so that its sum diverges, but so that it is square summable.
The reason for this will become clear from the bound featuring in Theorem~\ref{theo:oracle_bound}.

If we subtract $p(x)$ from both sides of~\eqref{def:recursion_one_step_estimator} and denote $\delta_k:=\hat p_k-p(x)$, then
\begin{equation}\label{eq:recursion_one_step_error}
\delta_{k+1} = (1-\gamma_k)\cdot\delta_k + \gamma_k\cdot \left\{ G_k - p(x) \right\}.
\end{equation}
From this identity a bound on the risk of the estimator~\eqref{def:hat_p_k} follows, under a specific condition to be imposed on the shape
of the connection function $\rho$.

\begin{theorem}[Oracle bound]\label{theo:oracle_bound}
Consider the estimator $\hat p_k$ of $p(x)$ defined in~\eqref{def:hat_p_k} and the sequence $\gamma_i$ defined in~\eqref{def:recursion_weight_and_increment}.
Assume that the distribution of feature points $\bm X$ and $\bm Y$  and the function $\rho$ are such that
\begin{equation}\label{eq:moment_condition}
\Big[\mathbb{E}\Big\{ \rho\big(\|\bm X - \bm Y\|\big) - p\big(\bm X\big) \Big\}^3\Big]^2 =
O\left[\!\!\left[ n \cdot \Big[ \mathbb{E}\Big\{ \rho\big(\|\bm X - \bm Y\|\big) - p\big(\bm X\big) \Big\}^2 \Big]^3 \right]\!\!\right].
\end{equation}
Then, for any $0\le k_0\le k \le n$, with $\sigma^2 := \mathbb{V}{\rm ar}\,P_{i, j}$, the following bound holds:
\begin{equation}\label{eq:oracle_bound}
\begin{aligned}
\mathbb{E}|\delta_{k+1}|^2 &\lesssim
\mathbb{E}|\delta_{k_0}|^2 \cdot \exp\Big(-2 \sum_{i=k_0}^{k}\gamma_i\Big) +
\left(\sum_{i=k_0}^{k}\gamma_i^2\right) + \Big(\sum_{i=k_0}^{k}\gamma_i\Big)^2 \cdot\\ &\qquad
\cdot\left[
n^{-1}\Big\{3+4 \sigma^2\cdot\log n\Big\} +
\mathbb{E}\max_{i\in V_{k+1}} \big|p(X_i)-p(x)\big|^2
\right],
\end{aligned}
\end{equation}
where the inequality holds up to a universal, multiplicative constant.
The expectation is taken under the model described in Section~\ref{sec:model}, where $p(x)$ is the connection probability at the origin.
\end{theorem}
\begin{proof}
The bound is obtained from~\eqref{eq:recursion_one_step_error}.
The remainder of the derivation can be found in Section~\ref{apx:proofs:oracle_bound} of the Appendix.
\end{proof}
Note that this bound is not asymptotic, but rather holds for any $0\le k_0\le k \le n$.
We mention that in principle this bound cannot be simplified without imposing further assumptions on the model.

The assumption~\eqref{eq:moment_condition} is trivial if, for example, $\rho$ does not depend on $n$.
Another thing to note is that since the differences in the condition belong to $[0,1]$, then a sufficient condition for~\eqref{eq:moment_condition} to hold is
\[
n \cdot \mathbb{E}\Big\{ \rho\big(\|\bm X - \bm Y\|\big) - p\big(\bm X\big) \Big\}^2 > \delta,
\qquad \delta>0.
\]

The choice of a sequence $\gamma_i$ (or equivalently of the weights $w_i$) and of $k$ can now be motivated from the point of view of minimising the upper bound in~\eqref{eq:oracle_bound}.
Clearly, the upper bound is at least as large as the first term on the right hand side, so the sum in the exponent should diverge.
However, this same sum scales the third term in the upper bound, entailing that the sum should grow {\it slowly} to infinity.
Finally, the second term in the upper bound should converge to zero (i.e., the sequence $\gamma_i$ should be square summable).

The second line in~\eqref{eq:oracle_bound} contains two terms.
The first should in general be negligible, while the second is an approximation term that accounts for the fact that the connection probability is not constant over all vertices in the network.
The latter term is of course model specific.
For example in the SBM, $p(X_i)=p(x)$ as long as vertex $i$ belongs to the same community as the origin; cf.~\eqref{eq:SBM_example_1_truth}.
Ideally one would take $k$ to be as large as possible as long as all the vertices in $V_{k+1}$ belong to the same community.
In general, the slower $p(x)$ changes, the larger $k$ can be taken.\par

% \section[Minimising the oracle bound]{Minimising the oracle bound}\label{sec:minimise_bound}

The weights $w_\ell$ are one of the parameters of our estimator $\hat p_k$.
The examples in~\eqref{eq:weight_examples} correspond to the following choices for $\gamma_i$:
\begin{equation}\label{eq:gamma_examples}
\gamma_i^{(1)} =
\frac{|V_{i+1}\backslash V_i|}{|V_{i+1}|}, \quad
\gamma_i^{(2)} = \frac1{i+1}, \quad
\gamma_i^{(3)} \equiv \gamma,\qquad
i = 0, \dots, n,
\end{equation}
for some small $0<\gamma\ll1$.
(Technically, the bound in~\eqref{eq:oracle_bound} does not apply to the sequence $\gamma_i^{(1)}$ since it is random.)
In general $\gamma$ is chosen to depend on $n$ and on the smoothness of $p(x)$.
For H\"older $\beta$-smooth functions one would consider \[\gamma = \gamma_n = (\log n)^{(2\beta-1)/(2\beta+1)} n^{-2\beta/(2\beta+1)};\] cf.~\cite{belitser2013online} for other examples of possible conditions on the smoothness of $p(x)$.

It  now becomes clear  why, in Figure~\ref{fig:SBM_example_1_MSE}, the sequence $w^{(3)}$ tends to outperform $w^{(2)}$ (and $w^{(1)}$) if $k$ is relatively small.
The sequence $w^{(2)}$ corresponds to equally down-weighing any contribution to the estimate that does not come from the degree of the origin.
On the other hand, the sequence $w^{(3)}$ corresponds to weighing the contribution of each annulus of vertices to the estimator, based on their distance to the origin;
since $p(x)$ is flat in a neighbourhood of the origin, this is advantageous.

For each model, one should have a reasonable idea of how the expectation in the second line of~\eqref{eq:oracle_bound} behaves (as a function of $k$, that is).
For our example from Section~\ref{sec:estimator}, $p(x)$ is given by~\eqref{eq:SBM_example_1_truth}, so that this expectation can be easily bounded.
With $j_i
%= C(X_i)
$ representing the community   vertex $i$ belongs to, we have
\[
n\cdot p(X_i) - n\cdot p(x) = (c_{j_i}-c_1)\cdot q + c_1\cdot p_1-c_{j_i}\cdot p_{j_i},
\]
which in absolute value is bounded from above by a constant.
The dominating terms in the upper bound~\eqref{eq:oracle_bound} are then the first two terms.
This explains how in Figure~\ref{fig:SBM_example_1_MSE} we can have different combinations of weights $w_\ell$ and $k$ that lead to estimates with a comparable precision level.

The oracle bound from this section is useful when the nature of the underlying graph is well understood.
Nonetheless it is convenient to have a data-driven choice for the parameters of the estimator.
In the next section, we propose a procedure to make this choice.

%%%%%%%%%%%%%%%%%%%%%%%%%%%%%%%%%%%%%%%%%%%%%%%%%%
\section[Choice of the Weights $w_\ell$ and the Distance $k$]{Choice of the Weights $w_\ell$ and the Distance $k$}\label{sec:choice_of_weights_and_k}

In this section we propose a numerical procedure to select the weights $w_\ell$ and the size of the neighbourhood $k$ (or the number of neighbours $m$) from the data.
The idea underlying this procedure is that the oracle bound~\eqref{eq:oracle_bound} provides various reasonable candidates for the sequence of weights $w_\ell$.
For each of these candidate sequences of weights we select the best $k$ in terms of minimising the MSE $\mathbb{E}\{\hat p_k - p(x)\}^2$;
the best sequence of weights attains the best MSE at the optimal choice (or \emph{oracle}) of $k$, and provides us with the best estimate of the local connection probability.

The parameter $k$ can be interpreted as a bandwidth parameter that trades off the squared bias and the variance of the estimator.
If $k=0$, then we get the empirical estimator $\hat p_0$ which is unbiased.
Increasing $k$ biases the corresponding estimate but, due to averaging over neighbours, the variance is reduced.

We propose a Monte Carlo cross-validation (MCCV) procedure to select $w_\ell$ and $k$.
(The procedure can be adjusted in the obvious way if one prefers to work directly with the number of neighbours $m$ instead.)
Given a graph $G$, and a subset of vertices $V\subset\{0,\dots,n\}$, define $G(V)$ to be the subgraph induced by the vertices in $V$.
Denote by $V^{\rm c}$ be the complement of $V$ in $\{0,\dots,n\}$.
Consider $\{V_i\}_{i=1,\dots,M}$, with $M\in\mathbb{N}$, a collection of independent random subsets of the vertices $\{1,\dots,n\}$.
Let $G_i := G(V_i\cup\{0\})$ and $G_i^{\rm c} := G(V_i^{\rm c}\cup\{0\})$, and note that each $G_i$ is independent of $G_i^{\rm c}$, and that all $G_i$ (resp., all $G_i^{\rm c}$) have the same distribution.
Finally, denote by $\hat{p}_{i,k}$ the estimates obtained from $G_i$, and by $\tilde{p}_{i,k}$ the (independent) estimates obtained from $G_i^{\rm c}$.

The MCCV criterion that we propose is simply an estimate of the MSE of our estimator:
\begin{equation}\label{eq:MCCV_criterium}
R(k, \{w_\ell\}) := \frac1M \sum_{i=1}^M \big\{\hat{p}_{i,k} - \tilde{p}_{i,0}\big\}^2.
\end{equation}
Since the two terms in each difference are independent and since $\tilde{p}_{i,0}$ is an unbiased estimator of $p(x)$, we have, for each $i=1, \dots, M$,
\begin{align*}
\mathbb{E}\big\{\hat{p}_{i,k} - \tilde{p}_{i,0}\big\}^2 =&\,
\mathbb{E}\big\{\hat{p}_{i,k} - p(x) + p(x) - \tilde{p}_{i,0}\big\}^2\\ =&\,
\mathbb{E}\big\{\hat{p}_{i,k} - p(x)\big\}^2 + \mathbb{E}\big\{\tilde{p}_{i,0}-p(x)\big\}^2.
\end{align*}
Since the estimators $\hat{p}_{i,k}$, $i=1,\dots,M$ (resp., $\tilde{p}_{i,k}$, $i=1,\dots,M$) have the same distribution, we find
\[
r(k, \{w_\ell\}) := \mathbb{E}\,R(k, \{w_\ell\}) =
\mathbb{E}\big\{\hat{p}_{1,k} - p(x)\big\}^2 + \mathbb{E}\big\{\tilde{p}_{1,0}-p(x)\big\}^2.
\]
This means that although $R(k, \{w_\ell\})$ is not an unbiased risk estimator, its bias (which is the variance of the empirical estimator) is independent of $k$ and of the weights $\{w_\ell\}$ (since the empirical estimator does not depend on the weights). As a consequence, we can use~\eqref{eq:MCCV_criterium} to pick $k$ and $\{w_\ell\}$ in a data-driven way via minimisation of the risk estimate.

The only drawback of this approach is that the risk estimator $R(k, \{w_\ell\})$ refers to the sample size of (roughly) $n/2$.
This means that since we expect the optimal (oracle) for $k$ to grow with $n$, a $\hat k$ that minimises $R(k, \{w_\ell\})$ should underestimate the oracle.
This should in general not be a concern though.
Firstly, the oracle will typically grow slowly with $n$, so that the bias of $\hat k$ should be relatively small.
Secondly, since in principle we have do not have insight into how the bias of the estimator changes with $k$, it is probably a good idea to underestimate $k$ anyway.

In Figure~\ref{fig:SBM_example_1_MCCV_reps_5050_n31} we depict the criterion~\eqref{eq:MCCV_criterium} as a function of the number of neighbours $m$. This we do for the graph from Figure~\ref{fig:SBM_example_1_estimates},  for different values of $M$, and for the three weight sequences that we introduced earlier.
To better compare the curves, these have been vertically shifted so that the location of the minima of all curves is at zero.
It is clear that even for a relatively small number of replications $M$, the criterion quickly stabilises, so that the location of the minimiser is quite insensitive to the number of repetitions.
This means that estimates of the optimal choice of $k$ (or $m$) can be obtained at a low computational cost.

\begin{figure}[!htb]
\centering
\includegraphics[trim={0cm 0cm 0cm 0cm}, width=0.99\textwidth]{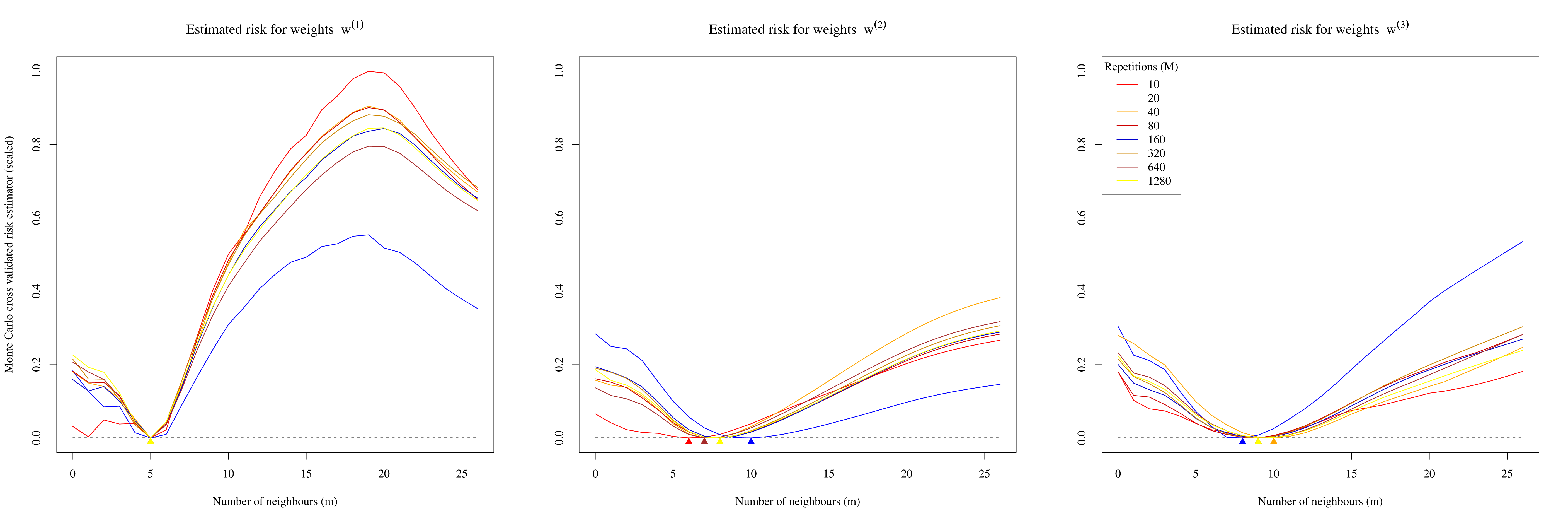}
\caption{\footnotesize
Some examples of the MCCV criterion~\eqref{eq:MCCV_criterium} for the graph from Figure~\ref{fig:SBM_example_1_estimates}.
The three plots correspond the three weight sequences.
In each plot we depict the criterion for a number of repetitions $M\in\{10 \cdot 2^i, i = 0, \dots, 7\}$.
The curves have been vertically adjusted so that the location of the minima of the curves is at zero.
The triangles below the x-axis indicate the locations of the corresponding minima.
}\label{fig:SBM_example_1_MCCV_reps_5050_n31}
\end{figure}

Figures~\ref{fig:SBM_example_1_best_k_cv} and~\ref{fig:SBM_example_1_best_ests_cv} exemplify the procedure for our running example.
Figure~\ref{fig:SBM_example_1_best_k_cv} shows the estimates of $m$ for different values of $n$, while Figure~\ref{fig:SBM_example_1_best_ests_cv} shows the resulting estimates for the connection probability $p(x)$.
We observe that the estimates of $m$ are relatively concentrated, and on average seem to stabilise and grow slowly with $n$.
As for the resulting estimates of $p(x)$, we see that in the majority of the cases there is a substantial improvement over the empirical estimator.
This can also be seen by comparing these plots with the first row of plots in Figure~\ref{fig:SBM_example_1_best_ests}.

\begin{figure}
%[!htb]
\centering
\includegraphics[trim={0cm 0cm 0cm 0cm}, width=0.80\textwidth]{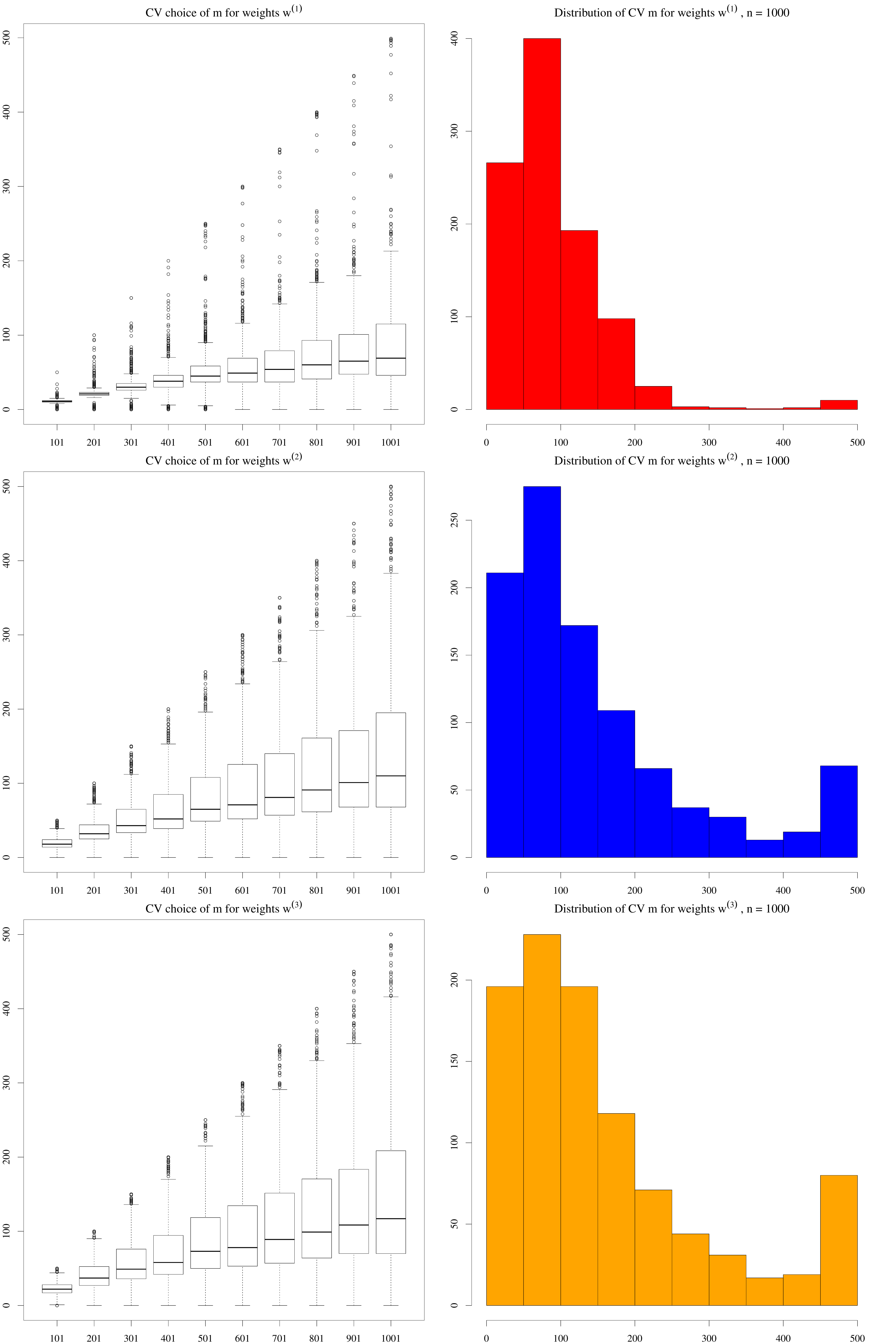}
\caption{\footnotesize
Estimates for $m$ for different sample sizes using our cross-validation procedure.
Each row of plots corresponds to a different weight sequence.
The boxplots in the left panels correspond to the estimates of $m$, for different choices of $n$.
In the right panels, each barplot shows the distribution of the estimates of $m$ for the largest value of $n$ that was considered.
}\label{fig:SBM_example_1_best_k_cv}
\end{figure}

\begin{figure}
%[!htb]
\centering
\includegraphics[trim={0cm 0cm 0cm 0cm}, width=0.80\textwidth]{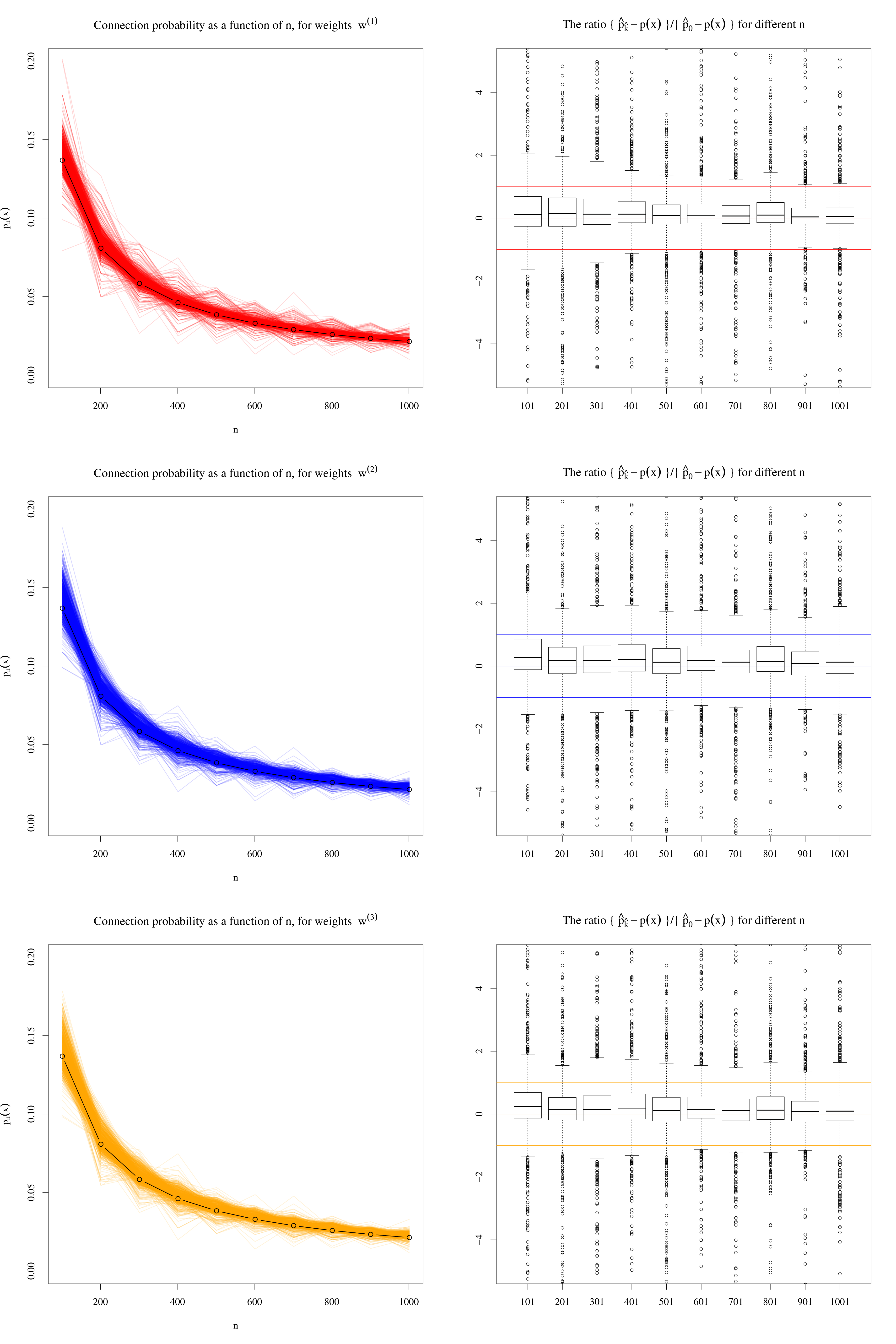}
\caption{\footnotesize
Estimates and errors corresponding to the estimates of $m$ for different sample sizes.
The left plots show the true connection probability (in black) for the different sample sizes, and the estimates for the three sets of weights in respectively grey, red, blue and orange.
On the right side we have boxplots of the errors for the different sample sizes.
These correspond to the ratio of the error estimates for each set of weights (with estimated $k$), divided by the error of the empirical estimator.
}\label{fig:SBM_example_1_best_ests_cv}
\end{figure}

% \newpage
%%%%%%%%%%%%%%%%%%%%%%%%%%%%%%%%%%%%%%%%%%%%%%%%%%
\section[Application to a Real Dataset]{Application to a Real Dataset}\label{sec:application}

In this section we apply our method to a real dataset\footnote{This dataset is part of the \emph{MaxMind WorldCities and Postal Code Databases}, and was obtained via \url{https://www.maxmind.com/en/free-world-cities-database}.}.
The dataset contains a listing of cities in the world, and their  corresponding population, latitude, and longitude.
We focused on some European cities.
We sampled $n=250$ cities (with replacement), with a probability proportional to their population.
We treated these as features $\bm X$, and chose the connection function $\rho(x)=\exp\{-(2/3)\cdot x\}$.
(The constant $2/3$ was chosen to ensure some sparsity on the resulting graph.)
Figure~\ref{fig:world_cities_map} depicts the location of the sampled cities marked with red squares, and the connections between them represented by a grey line; the city of Madrid, Spain, whose connection probability we estimate is marked blue\footnote{The figure was generated using R's \emph{rworldmap} package~\cite{south2011rworldmap}.}.
Figure~\ref{fig:world_cities_graph} represents the underlying graph, with each vertex containing the flag of the corresponding country.

\begin{figure}[!htb]
\centering
\includegraphics[trim={13cm 10cm 10cm 10cm}, width=0.80\textwidth]{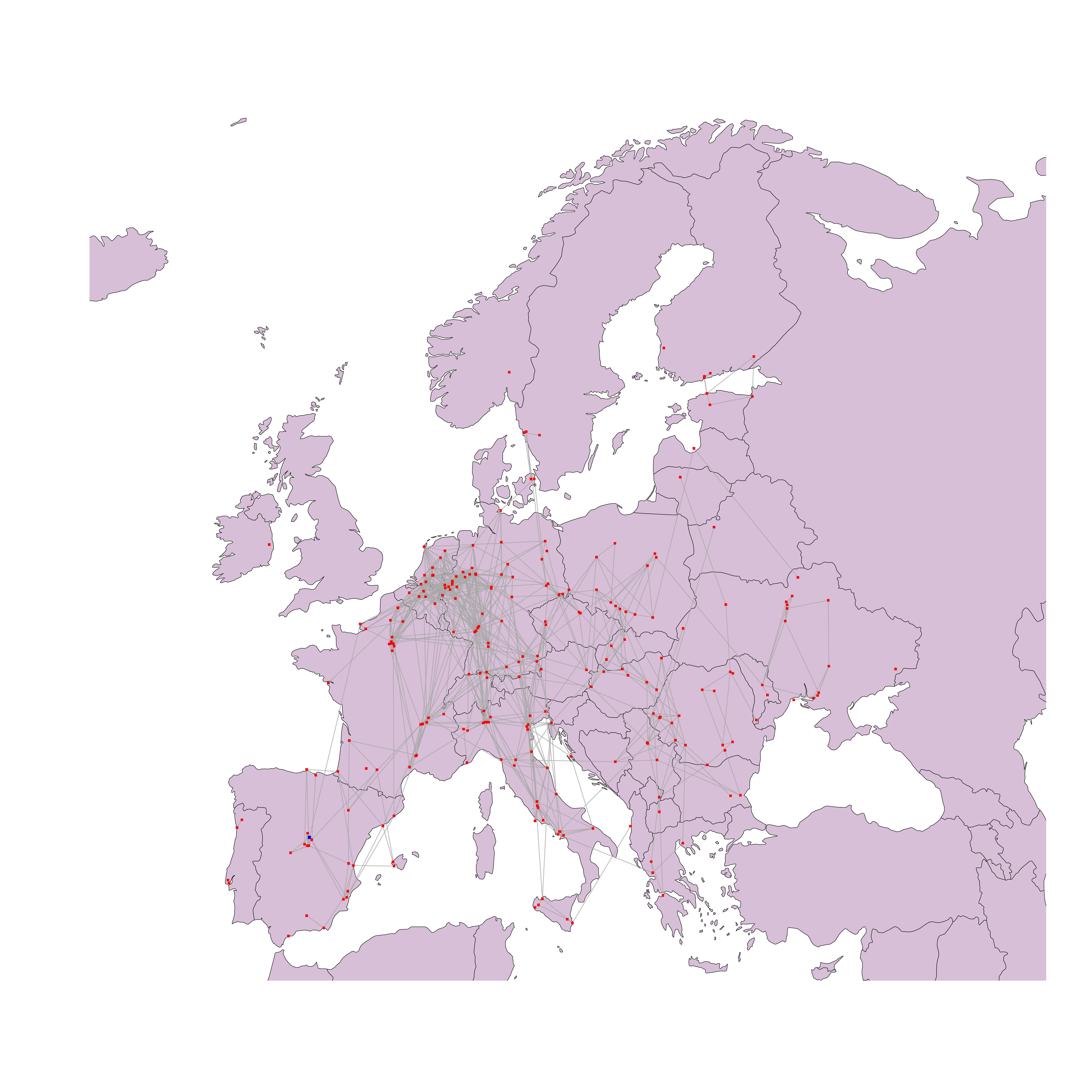}
\caption{\footnotesize
Map indicating the locations of the sampled European cities and connections between them.
Red squares mark the locations of the cities, and grey lines indicate connections between cities.
The blue square marks the location of Madrid.
}\label{fig:world_cities_map}
\end{figure}

\begin{figure}[!htb]
\centering
\includegraphics[trim={0cm 0cm 0cm 0cm}, width=0.99\textwidth]{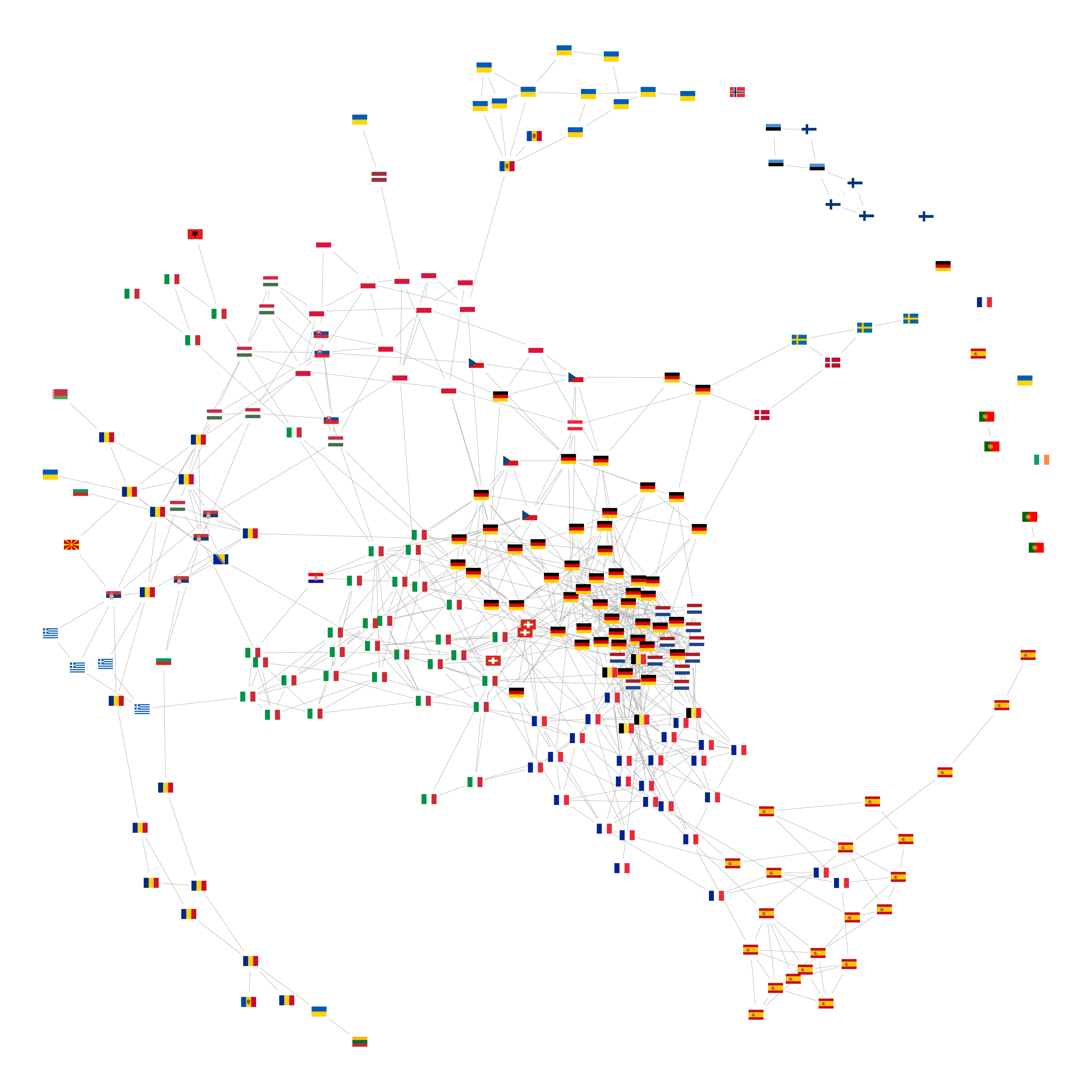}
\caption{\footnotesize
Graph of the sampled cities and corresponding connections.
Each vertex contains the flag of the country where the city lies.
}\label{fig:world_cities_graph}
\end{figure}

We applied our approach to the graph in Figure~\ref{fig:world_cities_graph} to estimate the connection probability of the city of Madrid.
The underlying connection probability if given by $p(x) = \mathbb{E}\,\rho(\|X-x\|)$.
In this case we compute this probability as
\[
p(x) = \frac1{15549}\sum_{i=1}^{15549}\rho(\|X_i-x\|),
\]
where $X_i$ represents the location of the $i$-th of the $15549$ European cities that we considered.
We computed $p(x)$, the connection probability for Madrid, as being $2.354872\times 10^{-2}$.
This connection probability tells us something about the probability that a newly chosen location in the same area would have of connecting to other vertices in the network.
We applied our estimation procedure to estimate this probability.
Figure~\ref{fig:world_cities_CV} depicts the MCCV criteria used to choose the number of neighbours $m$, based on $M=1000$ repetitions.
The estimated value for $m$ corresponded to considering respectively $85$, $126$, and $16$ neighbours.

\begin{figure}[!htb]
\centering
\includegraphics[trim={0cm 0cm 0cm 0cm}, width=0.99\textwidth]{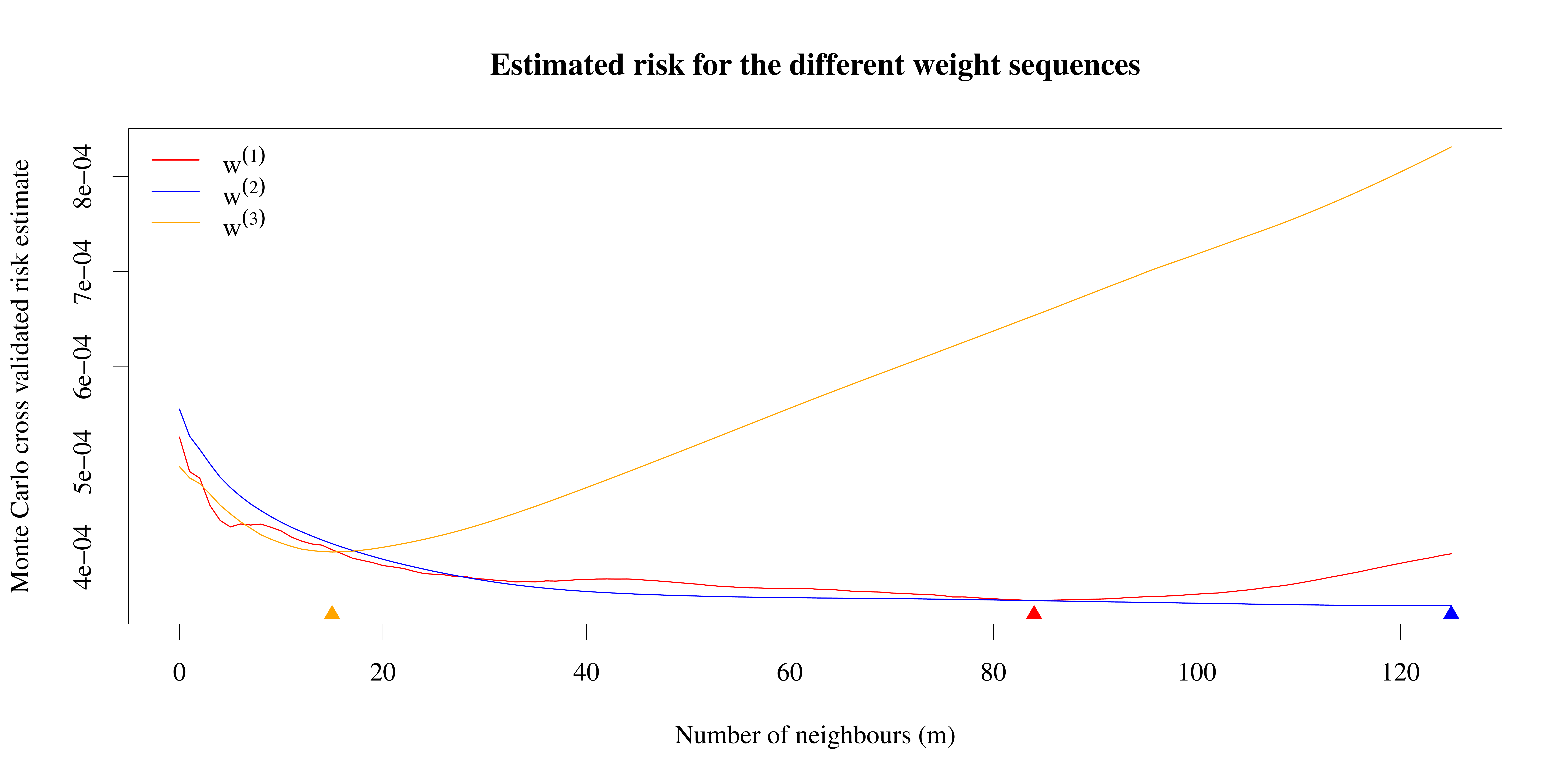}
\caption{\footnotesize
MCCV criteria for selecting the size of the neighbourhoods.
Each curse corresponds to a different set of weights.
The location of the minima are indicated by triangles.
}\label{fig:world_cities_CV}
\end{figure}

The corresponding estimates of the connection probability for the three sets of weights are respectively
$3.745882\times 10^{-2}$,
$3.391109\times 10^{-2}$, and
$2.338469\times 10^{-2}$; the empirical estimate was
$2.4\times 10^{-2}$.
The corresponding relative errors are respectively
$59.07\%$,
$44.00\%$,
$ 0.70\%$, and
$ 1.92\%$.
The conclusion is that the estimator based on the third set of weights, $w_\ell^{(3)}$, provides the best performance.
In particular, it substantially improves the empirical estimator.
The two other sets of weights underperform, which is perhaps not surprising.
The connection probability $p(x)$ should vary smoothly as the location of the feature $x$ changes making the second set of weights more appropriate.
The two other sets of weights do not seem to correctly capture the variation of $p(x)$.
Although these two sets of weights can reach a lower risk, the minimum in the MCCV criterion curve for the third set of weights is much more pronounced.
Also, as we have seen before, neighbourhoods sizes that are too large can correspond to an artificial minimiser of the risk; cf.~Figure~\ref{fig:SBM_example_1_best_k}.
In general one should expect the neighbourhood size to grow roughly logarithmically with the number of vertices.

\section[Application to Simulated Data]{Application to Simulated Data}\label{sec:simulations}
%%%%%%%%%%%%%%%%%%%%%%%%%%%%%%%%%%%%%%%%%%%%%%%%%%

\subsection[Mobile ad-hoc wireless networks]{Mobile ad-hoc wireless networks}\label{sec:simulations:ah_hoc_wireless_network}

In this section we exemplify how our results can be used in the context of the performance evaluation of mobile ad-hoc wireless networks.
Such networks are a collection of two or more mobile devices that are equipped with wireless communication and networking capabilities~\cite{haas2002wireless}.
These devices can communicate {\it directly} with other devices if these  are within a given range. In addition, they can {\it indirectly} communicate with the devices outside this range as long as there is a path along which the message can be relayed (which is possible in case each pair of subsequent devices along the path can communicate directly with each other).
Given the nature of the devices, it is reasonable to assume that the exact locations of the devices are not known.
Instead, we only assume to have access to a snapshot of the network in which it was recorded which pairs of devices communicated within a certain time window.
The problem is then to assess from such snapshots how likely it is that devices at a certain location (particularly with low coverage) can access another device in the network.

We can model a snapshot of such network using an RCM,
where individual devices are the vertices in the network.
These are sampled independently from some fixed distribution.
An edge is present between two vertices if these are within a certain distance of one another;
this means that the connection function is an indicator of the distance between vertices.

In our simulation the design point density is a mixture of three Gaussian measures on $\mathbb R^2$, restricted to the square $[0,10]^2$.
This is supposed to represent the distribution of mobile devices in a certain area of interest.
The means of the three mixture components were
\[
\mu_1 = (9,9), \qquad
\mu_2 = (8,3), \qquad\text{and}\qquad
\mu_3 = (3,9).
\]
The variance-covariance matrices were
\[
\Sigma_1 = \begin{bmatrix}4 & 1.2\\ 1.2 & 4 \end{bmatrix}, \qquad
\Sigma_2 = \begin{bmatrix}4 & 0\\ 0 & 4 \end{bmatrix}, \qquad\text{and}\qquad
\Sigma_3 = \begin{bmatrix}4 & 2\\ 2 & 4 \end{bmatrix}.
\]
The mixture weights were $w_1=0.4$, $w_2=w_3 = 0.3$.
The origin was placed at $(3,3)$, which corresponds to a relatively low value of the density, and the connection function was $\rho(x)=1_{\{x\le 2\}}$.
Figure~\ref{fig:wireless_example} displays a heatmap of the density of the design points, with a graph sampled from the random connection model being overlaid;
the red circle indicates the neighbourhood of the origin at which edges are established.

\begin{figure}[!htb]
\centering
\includegraphics[trim={4cm 7cm 3cm 2cm}, width=0.99\textwidth, clip]{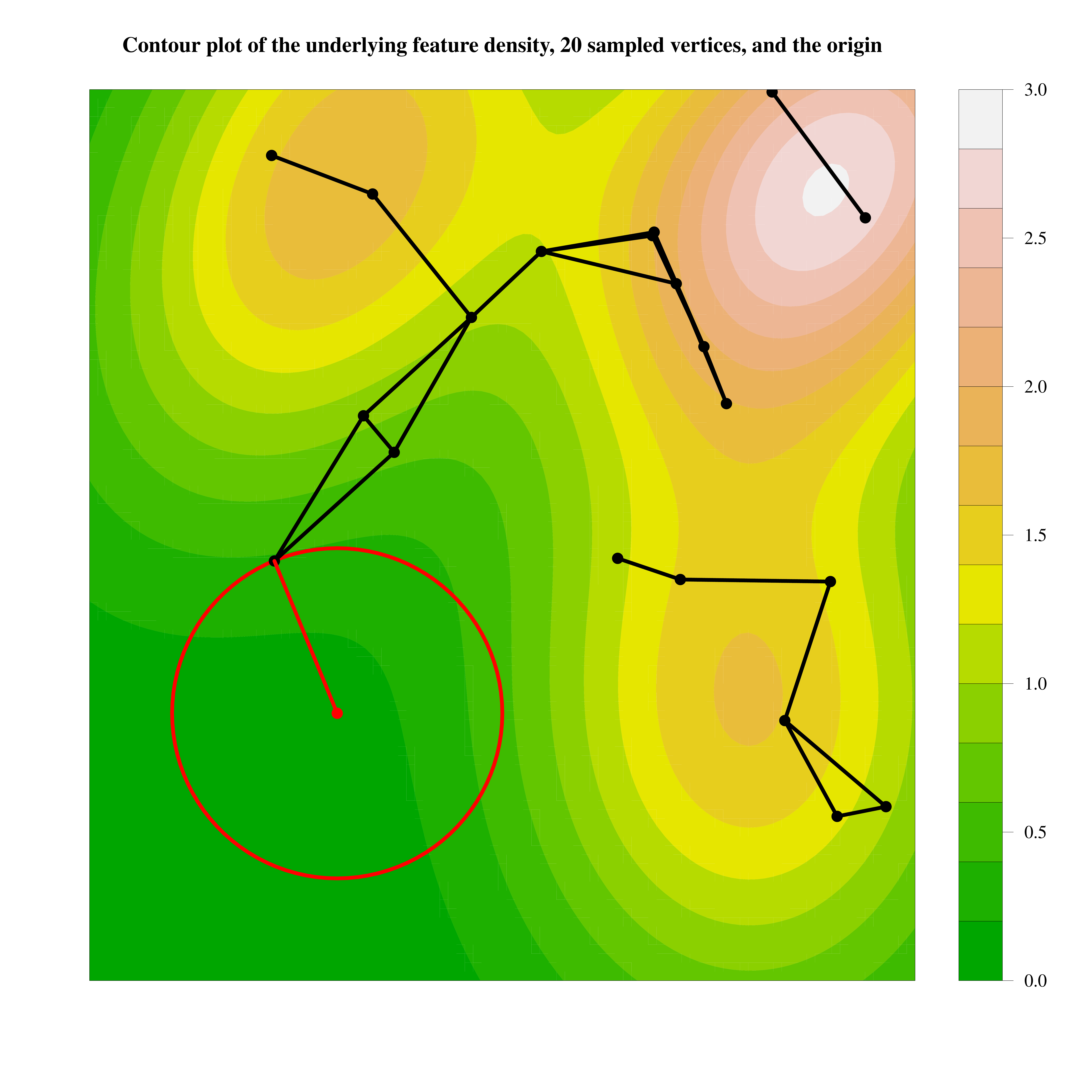}
\caption{\footnotesize
Design density and sampled network.
The heatmap represents the density of the design points.
The black dots represent the sampled vertices, while the red dot represents the origin which was placed in a low density region of the network.
The red circle represents the range of the connection function for the origin.
The lines connecting pairs of vertices are the edges of the network.
}\label{fig:wireless_example}
\end{figure}

The question that we ask is the following.
\emph{How many vertices should be present in the network so that the origin connects to at least one more vertex in the network with probability at least $0.9$?}
The answer is simple, but depends on the (unknown) connection probability  of the origin, i.e., $p=p(3,3)$.
In fact, if there are $n$ vertices in the network, the number of neighbours of the origin $B_0$ is distributed ${\rm Bin}(n,p)$.
If we want $\mathbb{P}(B_0>0)\ge0.9$, then we should pick
\[
n > n_0 := \left\lceil{\frac{\log(1-0.9)}{\log(1-p)}}\right\rceil.
\]
The quantity
\[
\bar n := \left\lceil{\frac{\log(1-0.9)}{\log(1-\hat p)}}\right\rceil.\]
is then our estimate for the minimal number of vertices that should be present in the network to ensure that the origin is connected to the rest of the network (through at least one vertex) with probability at least $0.9$.

The precision of $\bar n$ as an estimate of $n_0$ depends on how many vertices are in the network.
Since the asymptotics of our estimators are driven by $n\,p$, and $p$ is fixed in our case, then we should only expect $\bar n$ to be close to $n_0$ when $n$ is large.
In this sense, the result can be useful to trim vertices from a network while ensuring that vertices in regions of the network with less coverage still connect to the remainder of the network with high probability.

In our simulation we looked at how the estimate $\bar n$ of $n_0$ evolves as a function of the number of vertices in the network.
We considered $n\in\{ 100\cdot i,\; i = 1,\dots, 50 \}$.
For each value of $n$ in this collection we independently sampled $100$ networks from the corresponding model.
For each dataset we then computed the estimates of the connection probability at the origin for our three sets of weights, as well as the corresponding estimates of $n_0$.
In all cases we estimated the neighbourhood size using the MCCV procedure from Section~\ref{sec:choice_of_weights_and_k} using $100$ replications.
The true connection probability at the origin, numerically approximated via Monte Carlo integration was $p=0.021745$, leading to $n_0 = 105$.

\begin{figure}[!htb]
\centering
\includegraphics[trim={0cm 0cm 0cm 0cm}, width=0.99\textwidth, clip]{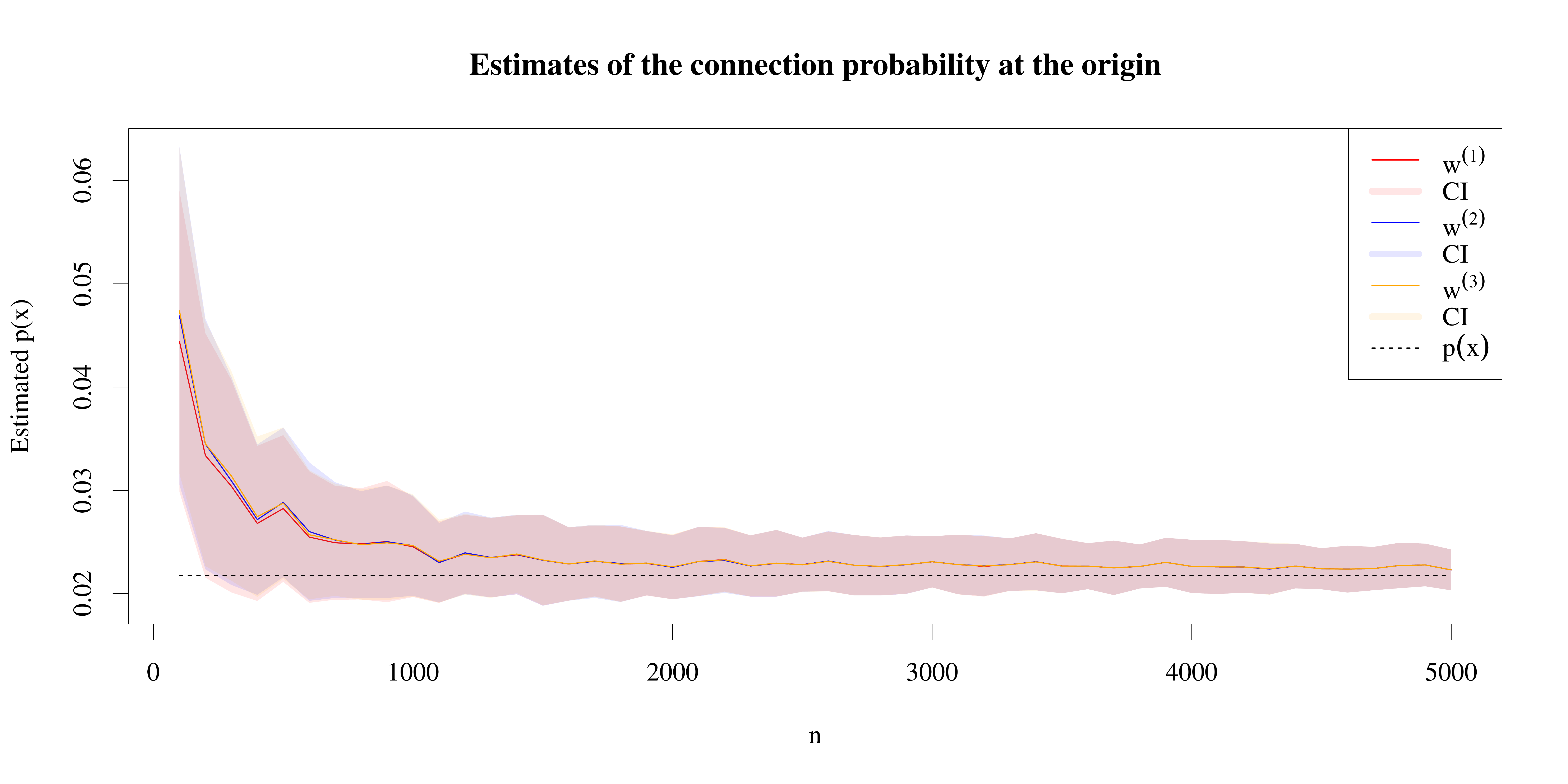}\\
\includegraphics[trim={0cm 0cm 0cm 0cm}, width=0.99\textwidth, clip]{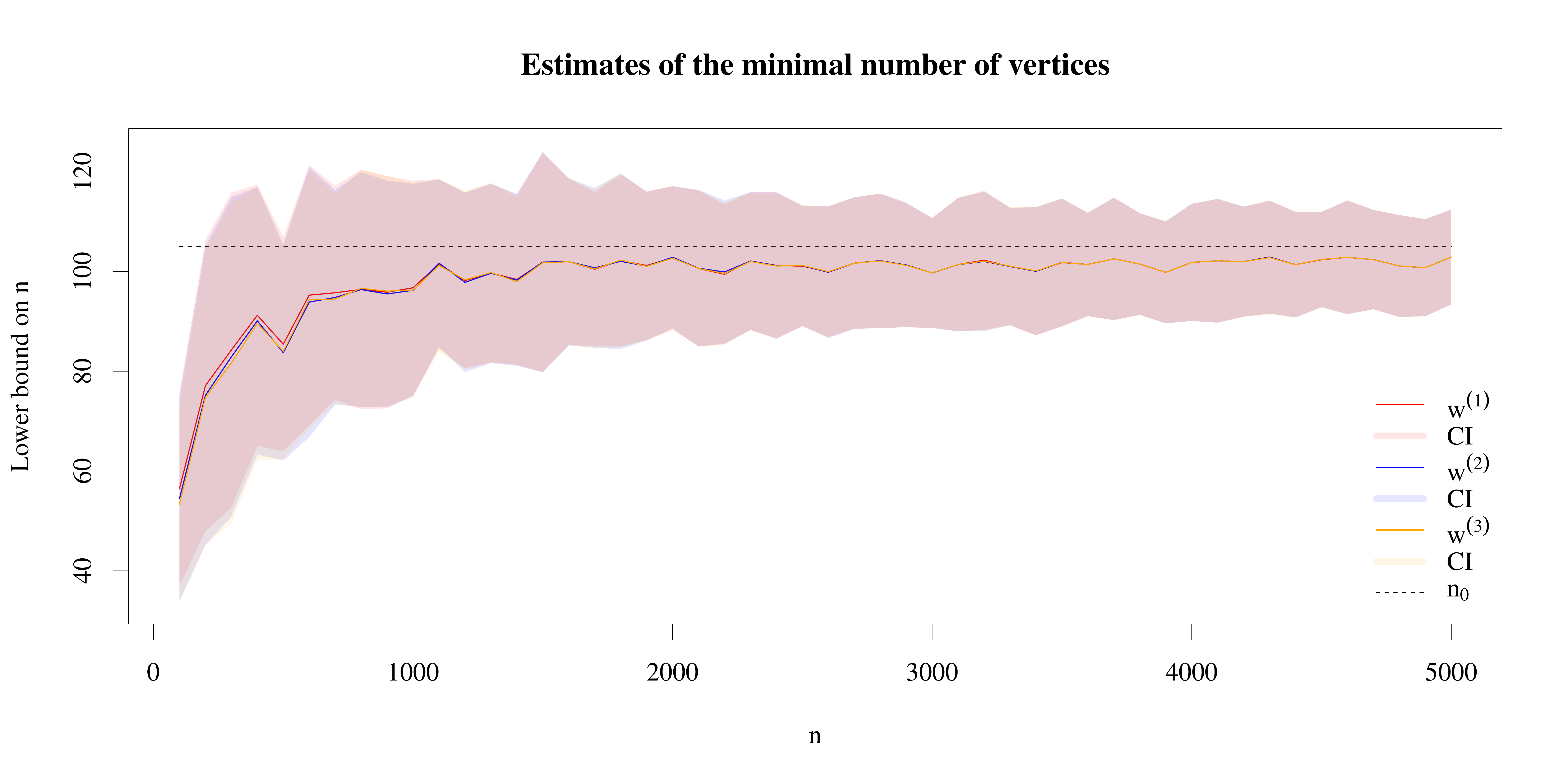}
\caption{\footnotesize
Top:
The mean of all estimates of the connection probability at the origin for our three sets of weights (full lines).
The shaded areas correspond to the mean of the estimates $\pm$ one standard deviations.
The dashed line marks the connection probability $p(x)$.
Bottom:
Mean of all estimates of $n_0$ for our three sets of weights (full lines).
the shaded areas correspond to the mean of the estimates $\pm$ $1.96$ standard deviations, and correspond to an approximate $95\%$ confidence interval.
The dashed line marks the the lower bound $n_0$.
}\label{fig:wireless_estimates}f
\end{figure}

Figure~\ref{fig:wireless_estimates} depicts the results of the simulation.
For for each combination of weight sequence and $n$, we generate $100$ estimates of $p$ and from those we obtain $100$ estimates of $n_0$.
The two plots show the evolution of the mean of the estimates of $p$ (on top) and of the estimates for $n_0$ (at the bottom) for the different values of $n$.
In both plots, and for each set of weights, the bands encompass  all values within one standard deviation of the mean of the corresponding  estimate.
The dashed horizontal lines mark the true values of the corresponding parameters.

We can see that for this particular combination of connection function and design point distribution the estimates are hardly affected by the choice of the weights.
As expected, the bias and variance of the estimator decrease as $n$ grows.
It also seems clear that the lower bound is underestimated, which follows as a consequence of the connection probability being overestimated.

To verify to what extent these estimates lead to the desired connection probability, we performed another Monte Carlo simulation.
(Since the results from Figure~\ref{fig:wireless_estimates} do not seem affected by the set of weights considered, we simply use the first set of weights in what follows.)
For each different value of $n\in\{ 100\cdot i,\; i = 1,\dots, 50 \}$ the bottom plot in Figure~\ref{fig:wireless_estimates} provides an estimate for $n_0$.
For each of these estimates $\hat n$ we produced $10^{5}$ graph with $\hat n$ vertices, and then proceed to check the fraction of those where the origin connects to any other vertex in the graph.
This fraction gives us a good approximation to the underlying connection probability at the origin for a graph with those many nodes.
Since the results in Figure~\ref{fig:wireless_estimates} are averaged over multiple runs (and therefore less variable), we also considered simply the estimates for each $n$ obtained in the first replication of the simulation.

\begin{figure}[!htb]
\centering
\includegraphics[trim={0cm 0cm 0cm 0cm}, width=0.99\textwidth, clip]{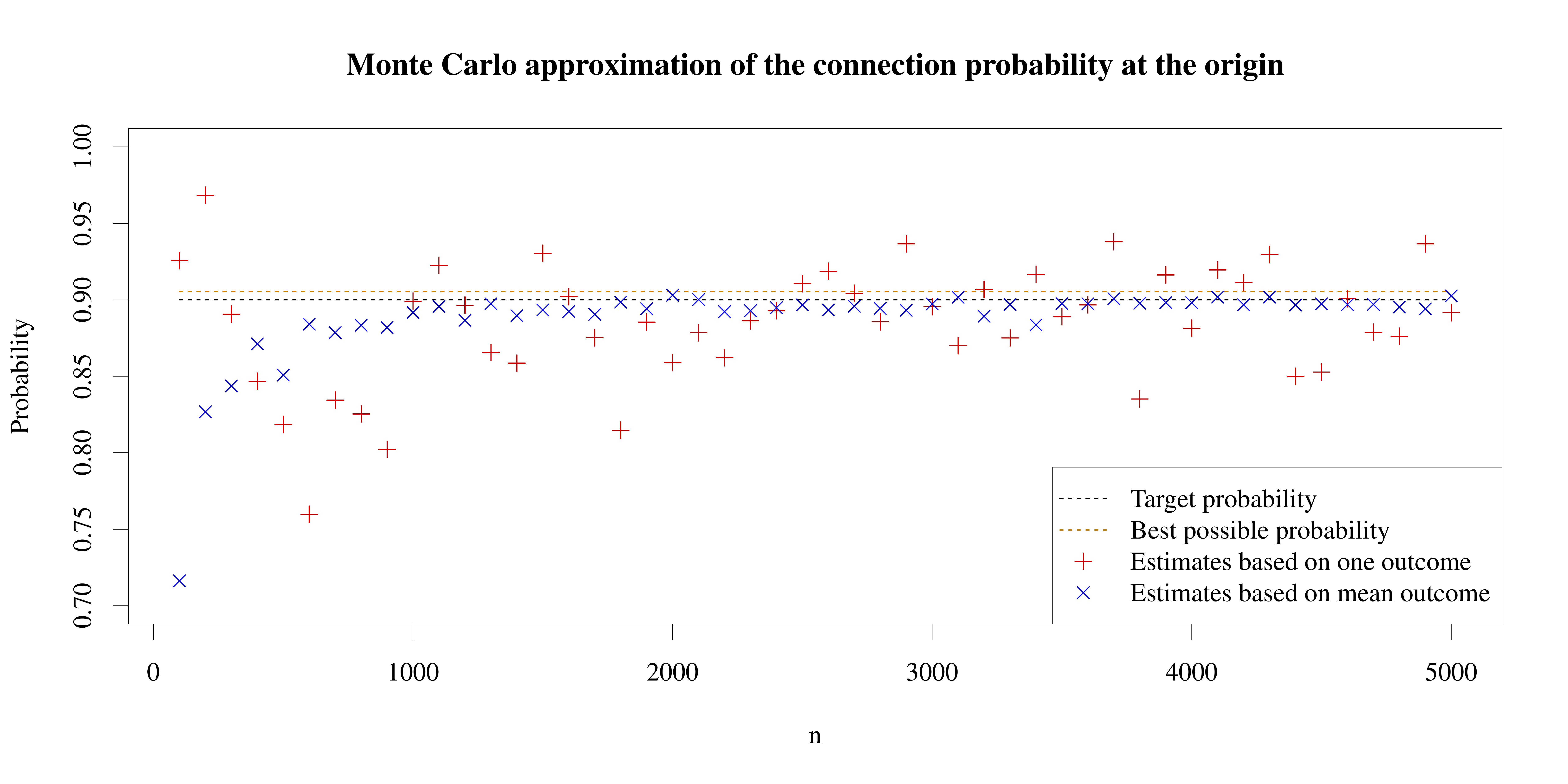}
\caption{\footnotesize
The horizontal black dashed line marks the target connection probability of $0.9$;
the horizontal orange dashed line marks the connection probability $0.9047$ in a graph with $n_0$ vertices.
The blue ${\rm x}$'s mark the probabilities based on the average estimates from Figure~\ref{fig:wireless_estimates};
the red $+$'s mark the probabilities corresponding to just one of the replication used to produce the same figure.
(All probabilities were approximated using a Monte Carlo simulation.)
}\label{fig:wireless_check}
\end{figure}

Figure~\ref{fig:wireless_check} contains the results of the simulation.
For reference, the figure also contains horizontal dashed lines marking the target probability $0.9$ and the connection probability $0.9047$ corresponding to a graph with $n_0$ vertices.
The approximations based on the mean estimated sample sizes appear to be, as one would expect, more stable than the approximations based on a single replication.
The latter should however be more representative of the results one would get by applying the method.
It is also clear that the quality of the estimates improves as the sample size grows.

As such, the method delivers an estimate of the number of vertices required to obtain a graph with the intended connection probability at the origin.
Note that the method delivers this based simply on a single graph where the origin is connected to a certain number of neighbours.
Therefore, most of the information about the connection probability at the origin is obtained from borrowing information from neighbours of the origin and exploring the underlying smoothness of the graph.

%%%%%%%%%%%%%%%%%%%%%%%%%%%%%%%%%%%%%%%%%%%%%%%%%%
\subsection[Different design distributions]{Different design distributions}\label{sec:simulations:design}

In this section we perform  more numerical simulations to illustrate how our estimator performs,
and particularly how it compares to the empirical estimator of the connection probability.

We picked three different distributions to sample features from.
For each of these we considered three different numbers of features.
Finally, for each combination of density and number of features we introduce the origin at three different locations corresponding to low, medium, and high connection probability.
In each of these settings we compute our estimates (for our three sets of weights), and pick the size of the neighbourhood using the MCCV method outlined in Section~\ref{sec:choice_of_weights_and_k}.

\begin{figure}[!htb]
\centering
\includegraphics[trim={0cm 0cm 0cm 0cm}, width=0.99\textwidth]{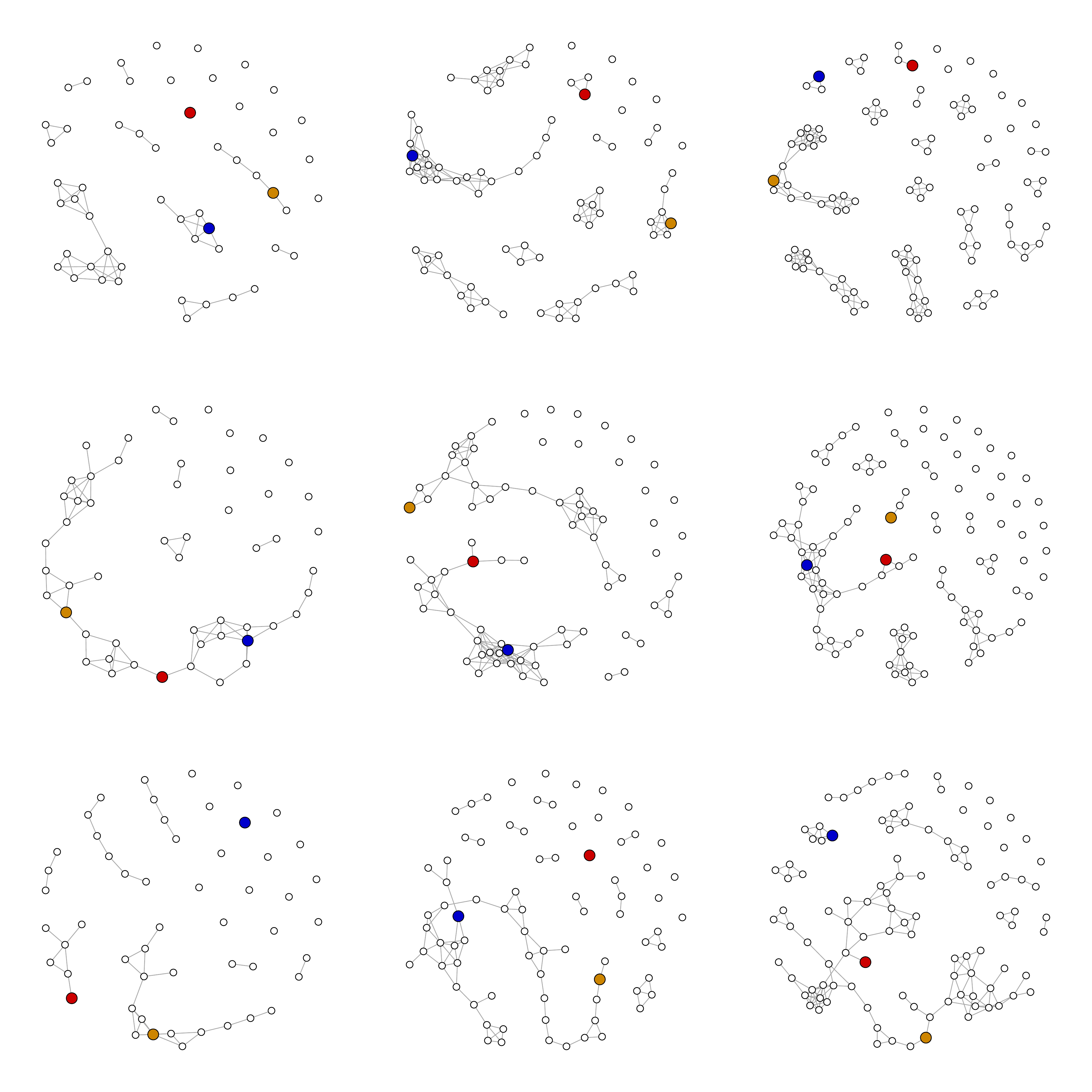}
\caption{\footnotesize
Examples of the graph used in the simulations.
Each row corresponds to a different choice of distribution for the features (top to bottom), and each column to a choice of number of features (left to right.)
Furthermore, for illustrative purposes, in each graph we introduced our three choices for the origin such that each graph contains $n+3$ vertices; in the simulations $n+1$ vertices are used (corresponding to $n$ features plus one origin.)
The origins are coloured red, blue, and orange for respectively, low, medium and high connection probability.
}\label{fig:SIM_example_graph}
\end{figure}

Figure~\ref{fig:SIM_example_graph} contains examples of the graphs that were generated.
The distributions for the features that we considered here were the following.
The first was a beta distribution with parameters $2$ and $5$, so that it has mean $2/7$.
The second was a mixture of two bivariate Gaussian distributions where
the first component is centred at $(0,0)$ and has variance-covariance matrix $\bm I_2$, and
the second component is centred at $(2.75, 2.75)$, and has variance covariance matrix with $1$'s in the diagonal and $0.75$'s elsewhere; the mixture weights were both $0.5$.
Finally, the third distribution is a uniform distribution on the unit cube $[0,1]^3$.

The number of features (excluding the origin) that were considered were $\{50, 75, 100\}$, and the connection function was $\rho(x) = \indic\{x \le \alpha \}$.
The thresholds $\alpha$ where taken to be respectively $\alpha = 0.01$, $\alpha_n = 0.6, 0.5, 0.4$, and $\alpha = 0.2$, for the three choices of feature distribution.
These values of $\alpha$ where chosen to ensure the resulting graphs were not too dense, leading to non-trivial choices for the number of neighbours.

As for the locations of the origin, these were respectively
$X_0 = 0.5, 0.1, 2/7$,
$X_0 = (1, 1), (0, 0), (1.75, 1.75)$, and
$X_0 = (0.5, 0, 0), (0.5, 0.5, 0), (0.5, 0.5, 0.5)$,
for the three choices of design distributions;
for each of these, the three points correspond to increasingly smaller values of the connection probability $p(x)$.
The values of the connection probabilities $p(x)$ and of $n\, p(x)$ can be found in Table~\ref{tab:truths_and_rates}.

\begin{table}[!htb]
\centering
\caption{\footnotesize
Value of the connection probability $p(x)$ (and of $n\, p(x)$ in parenthesis), for different combinations of feature distribution, number of design points $n$, origin location.
}
{\footnotesize
\begin{tabular}{l|c|ccc|}
\cline{2-5}
& $n$   & Low & Medium & High \\ \hline
\multicolumn{1}{|l|}{} & 50  &2.83$\cdot 10^{-7}$ (1.44$\cdot 10^{-5}$) &2.51$\cdot 10^{-2}$ (1.28) &5.93$\cdot 10^{-2}$ (3.03) \\
\multicolumn{1}{|c|}{Beta} & 75  &2.06$\cdot 10^{-7}$ (1.05$\cdot 10^{-5}$) &2.35$\cdot 10^{-2}$ (1.20) &5.56$\cdot 10^{-2}$ (2.84) \\
\multicolumn{1}{|l|}{} & 100  &1.62$\cdot 10^{-7}$ (8.26$\cdot 10^{-6}$) &2.24$\cdot 10^{-2}$ (1.14) &5.30$\cdot 10^{-2}$ (2.70) \\ \hline
\multicolumn{1}{|c|}{Gaussian} & 50  &3.15$\cdot 10^{-2}$ (1.61) &4.45$\cdot 10^{-2}$ (2.27) &5.50$\cdot 10^{-2}$ (2.80) \\
\multicolumn{1}{|c|}{mixture} & 75  &2.97$\cdot 10^{-2}$ (1.51) &4.21$\cdot 10^{-2}$ (2.15) &5.24$\cdot 10^{-2}$ (2.67) \\
\multicolumn{1}{|l|}{} & 100  &2.80$\cdot 10^{-2}$ (1.43) &4.00$\cdot 10^{-2}$ (2.04) &5.00$\cdot 10^{-2}$ (2.55) \\ \hline
\multicolumn{1}{|l|}{} & 50  &6.13$\cdot 10^{-3}$ (3.13$\cdot 10^{-1}$) &1.20$\cdot 10^{-2}$ (6.12$\cdot 10^{-1}$) &2.35$\cdot 10^{-2}$ (1.20) \\
\multicolumn{1}{|c|}{Uniform} & 75  &5.63$\cdot 10^{-3}$ (2.87$\cdot 10^{-1}$) &1.10$\cdot 10^{-2}$ (5.63$\cdot 10^{-1}$) &2.17$\cdot 10^{-2}$ (1.11) \\
\multicolumn{1}{|l|}{} & 100  &5.18$\cdot 10^{-3}$ (2.64$\cdot 10^{-1}$) &1.02$\cdot 10^{-2}$ (5.19$\cdot 10^{-1}$) &2.00$\cdot 10^{-2}$ (1.02) \\ \hline
\end{tabular}}
\label{tab:truths_and_rates}
\end{table}

For each combination of design distribution, location of the origin, number of features, and weights sampled  $1\,000$ graphs from the corresponding RCM, and computing the logarithm of the ratio of the absolute error of our estimator with a cross-validated choice for the number of neighbours ($M=100$ in the MCCV procedure) and the absolute error of the empirical estimator.
The results of the simulations are summarised in Figure~\ref{fig:SIM_AE_ratio}.

\begin{figure}[!htb]
\centering
\includegraphics[trim={0cm 0cm 0cm 0cm}, width=0.99\textwidth]{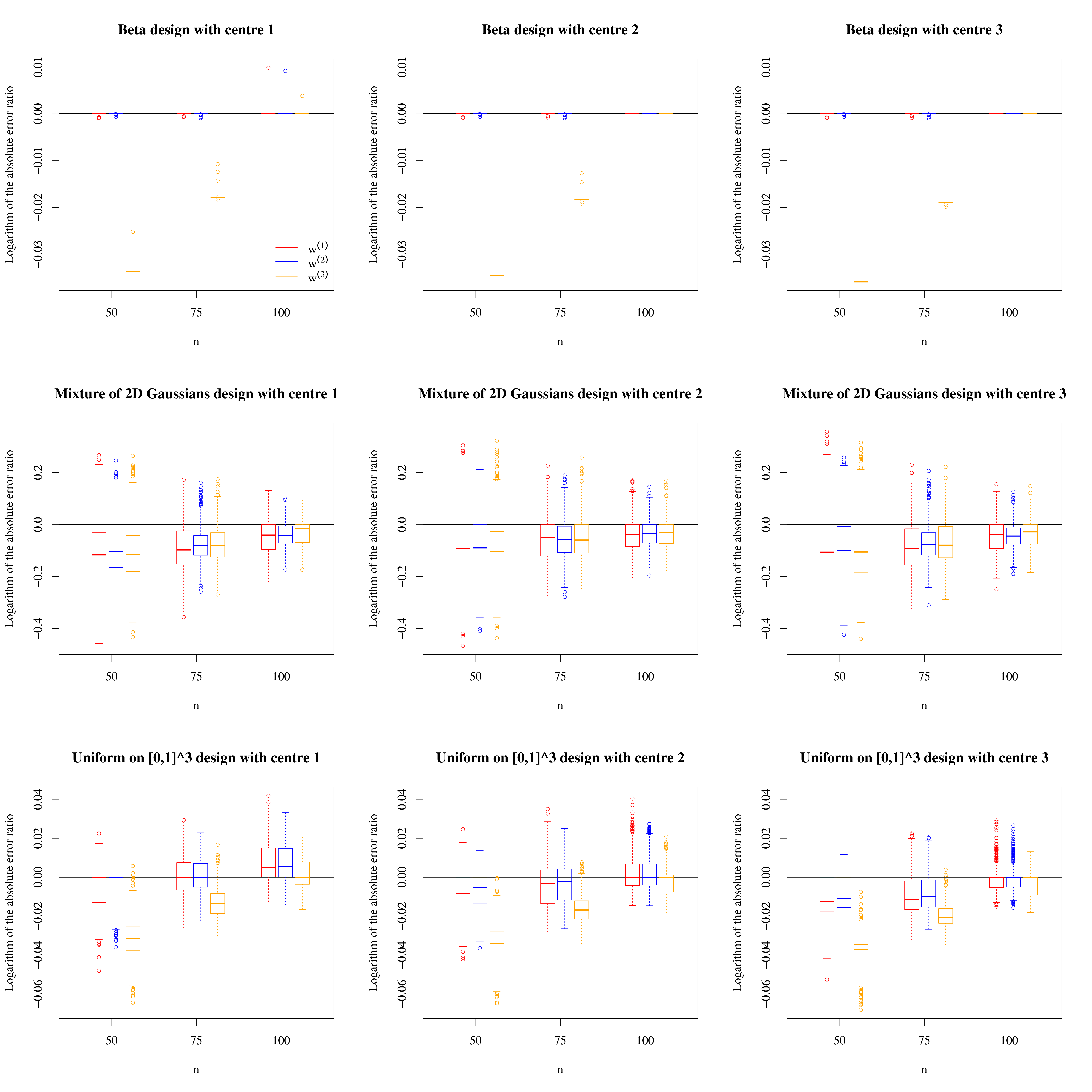}
\caption{\footnotesize
Results of the simulations.
Each row of plots corresponds to a different choice of distribution for the features (top to bottom), and each column of plots to a choice of number of features (left to right.)
In each plot, each set of three boxplots corresponds to the logarithm of the ratio of the absolute error of our estimator with cross validated choice of number of neighbours, and the absolute value of the error of the trivial estimator for the given number of features.
Finally, within each set of three boxplots, each one corresponds to a different set of weights for the estimator.
(Lower values indicate lower performance of the empirical estimator.)
}\label{fig:SIM_AE_ratio}
\end{figure}

First of all we remind that the asymptotics of the estimators are determined by $n\, p(x)$.
For the particular choices of the parameters that we made, one expects then the estimators to perform better for smaller values of $n$.
Also, for most of the sampled graphs, our estimator with the number of neighbours chosen via cross-validation perform better, irrespectively of the weights that are used.
However, it is the third set of weights that performs the best in almost all cases.
This is most likely related to fact that one expects the function $p(x)$ to vary smoothly as the location of the origin, $x$, changes.
The advantage of using the third set of weights is more evident when using the beta distributed features.

%%%%%%%%%%%%%%%%%%%%%%%%%%%%%%%%%%%%%%%%%%%%%%%%%%
\section[Conclusions]{Conclusions}\label{sec:conclusions}

Most interesting networks that are encountered in practice tend to be large, and typically display some form of inhomogeneity.
This inhomogeneity motivates why it makes sense to study the local counterparts of commonly studied metrics (such
as the degree distribution, the clustering coefficient, and the average path length).
In this paper we focused on the estimation of the local degree distribution of a given vertex of interest (termed the origin) in a network.
Equivalently, our objective is to estimate the local connection probability of said vertex;
this is the probability that the vertex under consideration establishes an edge with another vertex in the network.
The origin may be a vertex already present in the network, or it can be a \emph{probe vertex} that is introduced into the network to study its local properties.

The readily available empirical estimator for the connection probability (the degree of the vertex divided by the total number of vertices) is the maximum likelihood estimator and can be easily seen to be asymptotically Gaussian. One can, however, improve upon it by borrowing information from neighbouring vertices.
We propose a flexible, locally weighted average estimator which depends on two parameters:
the set of weights, and the set of neighbours to borrow information from.

To analyse the performance of the estimator we model the underlying network as arising from a random connection model.
This means that that we assume that there is a latent feature (potentially living in a high-dimensional metric space) associated with each vertex.
We then assume that the probability that two vertices connect is determined by a (perhaps unknown) {connection function} applied to the distance between the corresponding features.
If the connection function has bounded support (or vanishing tails), then with full (resp., high) probability, vertices only connect to their `neighbours' in the feature space.
Since features are assumed to be sampled from some arbitrary distribution, the resulting network will indeed be inhomogeneous.
The model that is thus constructed  is extremely flexible and captures several particular random graph models as particular cases.
It therefore offers a realistic modelling assumption for inhomogeneous networks.

Under the above modelling assumption we derive an oracle inequality that bounds the error of the estimator in terms of its parameters.
Under some prior knowledge on how the connection probability varies throughout the graph, the oracle inequality can be used to motivate the choice of weights (featuring in  the estimator).

To select which neighbours we borrow information from, we designed an MCCV procedure.
The set of vertices (other than the origin) is evenly partitioned into two disjoint sets.
Each set, together with the origin, is used to induce a network.
From one network we derive our estimator (as a function of the number of neighbours), while from the other we derive the empirical estimator.
The squared difference between the two is averaged over several random choices of the sets of vertices to deliver a function of the number of neighbours that we use as a data-driven criterion to pick the number of neighbours.
The procedure helps trading off the reduction in variance resulting from the averaging procedure, against the increase in bias resulting from borrowing information from neighbours.
Our numerical experiments show that the minimiser of the criterion stabilises rather quickly as the number of replications of the random sets grows, making this a  stable criterion for selecting the relevant neighbours of the origin.

Our numerical experiments further show that our estimator typically delivers more precise estimates of the local connection probability than the empirical estimator.
The precision of the estimator is not determined, as one might expect, by the number of vertices $n$ of the graph, but instead by $n\, p(x)$, i.e., by the expected degree of the origin.
Note that $p(x)$ itself may depend on $n$, so that the asymptotics of the estimator are perhaps better understood in terms of how $n\, p(x)$ grows with $n$.
If $n\,p(x)$ is too large, then the eccentricity of the origin is likely to be small in which case the network loses its inhomogeneity.
If the number of neighbours is selected in a data-driven way, then the estimator still performs as intended, but this is not an interesting regime for us since the local connection probability just becomes the (global) connection probability.
If, on the other hand, $n\, p(x)$ grows slowly with $n$, then the graph should display some local features and the problem is meaningful.
In this case we see that the selected number of neighbours grows slowly with $n$, again supporting the idea that indeed the network has local features.

%%%%%%%%%%%%%%%%%%%%%%%%%%%%%%%%%%%%%%%%%%%%%%%%%%
\section[Future Work]{Future Work}\label{sec:future}

The MCCV procedure delivers reasonable estimates for the size of the neighbourhoods that lead to estimates that offer an improvement over the empirical estimator.
However, the performance of our estimator is (of course) inferior to that of the oracle choice for the size of the neighbourhood.
It would be interesting to look into alternative methods to select the size of the neighbourhood, particularly methods that are less generic  (i.e., more tailored towards our specific setting).

It would also be of interest to provide more detailed results about the behaviour of the oracle of the size of the neighbourhood in terms of the parameters of the network.
(Interestingly, in our setting, the oracle is actually random.)
We suspect however, that this should only be possible by making much more specific assumptions on the structure of the network.

Another interesting related problem is that of community detection.
If the underlying graph has a community structure, and if these communities are different enough, then the way that the connection probability varies over vertices should provide information about when one transitions between communities.
Although we did investigate this possibility, in it's current form, the method is not computationally competitive with other community detection algorithms, but a modification of the approach tailored specifically for this purpose may improve upon this.

It would also be interesting to extend this approach to other features of networks, e.g., the clustering coefficient, or the distribution of the length of the shortest path.
The idea would be to estimate these from a subgraph induced by a neighbourhood of the origin.
The neighbourhood would be chosen as large as possible so that these features are still essentially constant over the neighbourhood.

\appendix
%%%%%%%%%%%%%%%%%%%%%%%%%%%%%%%%%%%%%%%%%%%%%%%%%%
\section{Proofs}\label{apx:proofs}

\subsection[Proof of the CLT for the trivial estimator]{Proof of the central limit theorem for the trivial estimator}\label{apx:proofs:CLT}

This CLT, which is a variation of the standard CLT, can be established in various ways; for completeness we include
an (elementary) proof.

Denote the sequence of random variables on the left hand side of~\eqref{eq:asymptotics_p_0} by
\[
Z_n := \frac{n \hat p_0 - p_n(x)}{\sqrt{n p_n(x) \{1-p_n(x))\}}},
\]
where $n \hat p_0 \sim {\rm Bin}\{n, p_n(x)\}$.
The moment generating function of $Z_n$ is
\[
\psi_n(t) =
\left[\{1-p_n(x)\}e^{-a_n(t)} + p_n(x)e^{b_n(t)}\right]^n,
\]
where
\[
a_n(t) = \frac{p_n(x) t}{\sqrt{n p_n(x) \{1-p_n(x))\}}}, \qquad\text{and}\qquad
b_n(t) = \frac{\{1-p_n(x)\} t}{\sqrt{n p_n(x) \{1-p_n(x))\}}}.
\]
Writing
\begin{align*}
e^{-a_n(t)} &= 1 - a_n(t) + \frac12 a_n(t)^2 + O\{a_n(t)^3\},\\
e^{b_n(t)}  &= 1 + b_n(t) + \frac12 b_n(t)^2 + O\{b_n(t)^3\},
\end{align*}
we obtain that, as $n\to\infty$,
\[
\psi_n(t) = \left\{1 + \frac1n c_n(t)\right\}^n \to e^{c(t)},
\]
as long as $c_n(t) \to c(t)$, where in our case,
\[
c_n(t) =
\frac n2\left[\{1-p_n(x)\}a_n(t)^2 + p_n(x)b_n(t)^2\right] +
O\{n a_n(t)^3\} + O\{n b_n(t)^3\}.
\]
(Note that the terms involving $a_n(t)$ and $b_n(t)$ cancel.)
By substituting $a_n(t)$ and $b_n(t)$, we conclude that as long as $n p_n(x)\to \infty$, then,  as $n\to \infty$,
\[
c_n(t) = \frac12t^2 + o(1) \to \frac12 t^2 = c(t).
\]
This shows  that $\psi_n(t) \to \exp(t^2/2)$, as $n\to\infty$.
By L\'evy's continuity theorem we then have that $Z_n$ is converges to a standard Gaussian random variable
as $n\to \infty$, under the proviso that
$n p_n(x)\to \infty$ as $n\to \infty$.

\subsection[Derivation of the oracle inequality~\eqref{eq:oracle_bound}]{Derivation of the oracle inequality~\eqref{eq:oracle_bound}}\label{apx:proofs:oracle_bound}

We conclude here the proof of Theorem~\ref{theo:oracle_bound} starting from the recursion in~\eqref{eq:recursion_one_step_error}.
From the features $\bm X$ from~\eqref{def:P_i_j}, and the random variables $(\epsilon_{i, j})_{i,j=0,\dots,n}$ from~\eqref{def:A_i_j}, we define the filtrations (indexed by $k$ and depending on $n$)
\[ %\begin{equation}\label{def:filtration}
\mathcal{F}_k^{(n)} :=
\sigma\big\{ \bm X, (\epsilon_{i, j}: i\in V_{k}, j=0,\dots,n) \big\}.
\] %\end{equation}
Note that $V_0,\dots,V_{k+1}\in\mathcal{F}_{k}^{(n)}$, so that $\hat p_{k}$ is measurable with respect to $\mathcal{F}_{k}^{(n)}$.
Using the fact that if $i\not\in V_k$, then \[\mathbb{E}\big[A_{i, j}\mid\mathcal{F}_{k}^{(n)}\big] = P_{i, j},\] it directly follows that
\[ %\begin{equation}\label{def:conditional_expectation}
g_{k} = \mathbb{E}\big[G_{k}\mid\mathcal{F}_{k}^{(n)}\big] =
\frac1{n|V_{k+1}\backslash V_{k}|}\sum_{i\in V_{k+1}\backslash V_{k}} \sum_{\substack{j=0\\ j\neq i}}^nP_{i, j}-\hat p_k.
\] %\end{equation}

Writing $D_{k} = G_{k} - g_{k}$,  it is clear that the $(D_k)_{k=0,\dots,n}$ are martingale increments with respect to the filtration $\{\mathcal{F}_{k}^{(n)}\}_{k=0}^n$.
We introduce the additional notation
\begin{align*}
%\label{def:conditional_bias_term}
\Delta_k &:=
\frac{\gamma_k}{|V_{k+1}\backslash V_{k}|}\sum_{i\in V_{k+1}\backslash V_{k}}
\bigg\{\frac1n\sum_{\substack{j=0\\ j\neq i}}^nP_{i, j} - p(X_i)\bigg\}, \\
%\label{def:bias_term}
\nabla_k &:=
\frac{\gamma_k}{|V_{k+1}\backslash V_{k}|}\sum_{i\in V_{k+1}\backslash V_{k}} \big\{p(X_i) - p(x)\big\}.
\end{align*}
It follows that we can rewrite the recursion~\eqref{eq:recursion_one_step_error} as
\begin{equation}\label{eq:recursion_one_step_error_decomposed}
\delta_{k+1} = (1-\gamma_k)\cdot\delta_k + \gamma_k\cdot D_k + \Delta_k + \nabla_k.
\end{equation}
Iterate this difference equation, we obtain for $k_0\in\{0,\dots,k-1\}$ (eventually depending on $n$),
\[
\delta_{k+1} =
\delta_{k_0}\cdot \prod_{i = k_0}^k(1-\gamma_i) +
\sum_{i=k_0}^k (\gamma_i \cdot D_i + \Delta_i + \nabla_i)\cdot\prod_{j=i+1}^k(1-\gamma_j).
\]
Now denote \[A_i = \sum_{j=k_0}^i \gamma_j\cdot D_j,\:\:\:B_i = \sum_{j=k_0}^i\Delta_j + \sum_{j=k_0}^i\nabla_j,\:\:\:H_i = A_i + B_i.\]
Using  `summation-by-parts', the summation in (\ref{eq:recursion_one_step_error_decomposed})  is
\[
\sum_{i=k_0}^k (H_i-H_{i-1})\cdot\prod_{j=i+1}^k(1-\gamma_j) =
H_k - \sum_{i = k_0}^{k-1} \gamma_{i+1} \cdot H_i \cdot  \prod_{j=i+2}^k(1-\gamma_j).
\]
Using  `summation-by-parts' one can also show that
\[
\sum_{i = k_0}^{k-1} \gamma_{i+1} \cdot  \prod_{j=i+2}^k(1-\gamma_j) =
1-\prod_{j=k_0+1}^k(1-\gamma_j) \le 1,
\]
which together with the triangle inequality and the inequality $1-x\le\exp(-x)$, $x\in\mathbb{R}$, can be used to obtain the following bounds:
\begin{align*}
|\delta_{k+1}| &\le
|\delta_{k_0}| \cdot \prod_{i = k_0}^k(1-\gamma_i) + |H_k| +
\sum_{i = k_0}^{k-1} \gamma_{i+1} \cdot |H_i| \cdot  \prod_{j=i+2}^k(1-\gamma_j)\\ &\le
|\delta_{k_0}| \cdot \exp\Big(-\sum_{i=k_0}^k\gamma_i\Big) + \max_{i=k_0,\dots,k}|H_i|\bigg\{ 1 +
\sum_{i = k_0}^{k-1} \gamma_{i+1} \cdot  \prod_{j=i+2}^k(1-\gamma_j)\bigg\}\\ &\le
|\delta_{k_0}| \cdot \exp\Big(-\sum_{i=k_0}^k\gamma_i\Big) + 2\max_{i=k_0,\dots,k}\Big|H_j\Big|.
\end{align*}
Squaring both sides, using the inequality $(x+y)^2 \le 2(x^2 + y^2)$, and taking expectations we arrive at the bound
\[
\mathbb{E}|\delta_{k+1}|^2 \le
2 \cdot \mathbb{E}|\delta_{k_0}|^2 \cdot \exp\Big(-2 \cdot \sum_{i=k_0}^k\gamma_i\Big) +
8 \cdot \mathbb{E}\max_{i=k_0,\dots,k}\Big|H_j\Big|^2.
\]
Since $(x+y+z)^2 \le 3(x^2 + y^2 + z^2)$, up to a constant multiple the expectation in the last term is no more than
\[
\mathbb{E}\max_{i=k_0,\dots,k}\Big|\sum_{j=k_0}^i \gamma_j\cdot D_j\Big|^2 +
\mathbb{E}\max_{i=k_0,\dots,k}\Big|\sum_{j=k_0}^i \Delta_j\Big|^2 +
\mathbb{E}\max_{i=k_0,\dots,k}\Big|\sum_{j=k_0}^i \nabla_j\Big|^2.
\]
It remains to bound these three terms.

We start with the first term.
Since the $D_j$ are martingale increments with respect to the filtration $\{\mathcal{F}_{k}^{(n)}\}_{k=0}^n$, the Burkholder maximal inequality~\cite{shiryaev1996probability} gives
\[
\mathbb{E}\max_{i=k_0,\dots,k}\Big|\sum_{j=k_0}^i \gamma_j\cdot D_j\Big|^2 \le
B_2 \cdot \sum_{j=k_0}^k \gamma_j^2\cdot \mathbb{E}D_j^2 \le
B_2 \cdot \sum_{j=k_0}^k \gamma_j^2,
\]
for a universal constant $B_2$.

As for the second term, for $m\in\{k_0,\dots,k\}$, it is
\[
%\sum_{l=k_0}^m \Delta_l =
\sum_{l=k_0}^m\frac{\gamma_l}{|V_{l+1}\backslash V_{l}|}\sum_{i\in V_{l+1}\backslash V_{l}}
\bigg\{\frac1n\sum_{\substack{j=0\\ j\neq i}}^nP_{i, j} - p(X_i)\bigg\} =
\sum_{l=k_0}^m\frac{\gamma_l}{n \cdot |V_{l+1}\backslash V_{l}|}\sum_{i\in V_{l+1}\backslash V_{l}}S_n^{(i)}.
\]
For any $i=0,\dots,m$, $S_n^{(i)}$ is the partial sum of the first $n$ terms of the sequence
\[
P_{i,0} - p(X_i),\; \dots,\; P_{i,i-1} - p(X_i),\; P_{i,i+1} - p(X_i),\; \dots.
\]
The terms in this sequence have mean 0 and are uncorrelated since for $j_1\neq j_2$,
\begin{align*}&
\mathbb{E}\{P_{i,j_1} - p(X_i)\}\{P_{i,j_2} - p(X_i)\} =
\mathbb{E}\mathbb{E}[\{P_{i,j_1} - p(X_i)\}\{P_{i,j_2} - p(X_i)\}\mid X_i]\\ &\qquad=
\mathbb{E}\left\{\mathbb{E}[P_{i,j_1} - p(X_i)\mid X_i]\mathbb{E}[P_{i,j_2} - p(X_i)\mid X_i]\right\}
=0,
\end{align*}
where we use that by definition $\mathbb{E}[P_{i, j}\,|\, X_i] = p(X_i)$.
Denote
\[
\sigma_n^2 =
\mathbb{V}{\rm ar}\,S_n^{(i)} =
\sum_{\substack{j=0\\ j\neq i}}^n\mathbb{V}{\rm ar}\left\{P_{i, j} - p(X_i)\right\} =
n\left\{\mathbb{E}[P_{i, j}^2] - \mathbb{E}[p(X_i)^2]\right\} =
n\cdot \sigma^2.
\]
(Note that $\sigma^2\le\mathbb{V}{\rm ar}\,P_{i, j}$.)
Since the $S_n^{(i)}$ are partial sums of a strictly stationary sequence of random variables and have variance $\sigma_n^2$, we know from~\cite{denker1986uniform} that
\[
\frac{S_n^{(i)}}{\sigma_n} \stackrel{\rm d}{\longrightarrow} N(0,1) \qquad\Leftrightarrow\qquad
\frac{(S_n^{(i)})^2}{\sigma_n^2} \quad \hbox{is uniformly integrable}.
\]

A sufficient condition for $(S_n^{(i)})^{2}/\sigma_n^2$ to be uniformly integrable is for it to have uniformly bounded moment of order larger than 1.
By the multinomial theorem,
\[
\mathbb{E}(S_n^{(i)})^{3} =
\sum_{|\bm\alpha|=3} {3 \choose \bm \alpha} \cdot
\mathbb{E}\big\{ P_{i,\cdot} - p(X_i) \big\}^{\bm \alpha}
\]
where $\bm\alpha$ is a multi-index of length $n$, and $P_{i,\cdot} - p(X_i)$ represents the vector $\{P_{i,0} - p(X_i), \dots, P_{i,i-1} - p(X_i), P_{i,i+1} - p(X_i), \dots P_{i,n} - p(X_i)\}$.
Since these terms are conditionally independent given $X_i$, and identically distributed,  all expectations in the previous display are equal to zero with the exception of $s := \mathbb{E}\{P_{i, j} - p(X_i)\}^3$.
From this we conclude that uniform integrability follows if $s$ is such that
\[
\mathbb{E}(S_n^{(i)})^{3}  = n\cdot s \le O\big\{n \cdot \sigma^2\big\}^{3/2}.
\]
This is ensured by assumption~\eqref{eq:moment_condition}.

Conclude that if $Z_n^{(i)} = S_n^{(i)}/n$, then $\sqrt n Z_n^{(i)}/\sigma \stackrel{\rm d}{\longrightarrow}N(0,1)$, so that if $Z$ is distributed as a $N(0,1)$ random variable, then, as $n\to\infty$,
\[
\mathbb{P}\left( \left|\sqrt n\frac{Z_n^{(i)}}{\sigma}\right| > M \right) \rightarrow \mathbb{P} (|Z|>M) \le 2e^{-\frac12M^2},\quad i=1,2,\dots,
\]
whence, $\mathbb{P} (|Z_n^{(i)}|>M) \le 3 \exp\big\{-n\cdot M^2/(2\sigma^2)\big\}$, $i=1,2,\dots$, for all large enough $n$.
Using the triangle inequality and monotonicity
\[
\mathbb{E}\max_{l=0,\dots,k} \left| \sum_{m=0}^l\frac{\gamma_m}{|V_{m+1}\backslash V_{m}|}\sum_{i\in V_{m+1}\backslash V_{m}} Z_n^{(i)}\right|^2	\le
\left(\sum_{m=0}^k\gamma_m\right)^2 \mathbb{E}\max_{i\in V_{k+1}}|Z_n^{(i)}|^2.
\]
Conditioning on $\{\max_{i\in V_{k+1}}|Z_n^{(i)}|>M\}$, we find
\begin{align*}
\mathbb{E}\max_{i\in V_{k+1}}|Z_n^{(i)}|^2 &\le
\mathbb{P}\left\{\max_{i\in V_{k+1}}|Z_n^{(i)}|>M\right\} + M^2\le
\sum_{i=0}^n\mathbb{P}\left\{|Z_n^{(i)}|>M\right\} + M^2\\ &\le
3\exp\left\{\log(n+1)-\frac n{2\sigma^2}M^2\right\} + M^2,
\end{align*}
for any $M\ge0$.
Setting $M^2 =4\cdot\sigma^2 \cdot \log(n+1)/n$, it follows that
\[
\mathbb{E}\max_{i\in V_{k+1}}|Z_n^{(i)}|^2 \le \frac{3+4\cdot\sigma^2}n\cdot\log(n+1).
\]

The third term is bounded using the triangle inequality and monotonicity:
\[
\mathbb{E}\max_{i=k_0,\dots,k}\Big|\sum_{j=k_0}^i \nabla_j\Big|^2 \le
\Big(\sum_{i=k_0}^k \gamma_i\Big)^2 \cdot \mathbb{E}\max_{i\in V_{k+1}}|p(X_i)-p(x)|^2.
\]
The statement of the theorem follows by combining the three upper bounds.

% \newpage
\bibliographystyle{plain}
\bibliography{references}

\end{document}